\documentclass[preprint,a4paper,12pt,number]{elsarticle}

\usepackage{latexsym}
\usepackage{amssymb}
\usepackage{amsmath}
\usepackage[normalem]{ulem}
\usepackage[all]{xy}

\newcommand{\ignore}[1]{} 

\newtheorem{theorem}{Theorem}
\newtheorem{lemma}[theorem]{Lemma}
\newtheorem{proposition}[theorem]{Proposition}
\newtheorem{corollary}[theorem]{Corollary}
\newdefinition{definition}[theorem]{Definition}
\newdefinition{example}[theorem]{Example}
\newproof{proof}{Proof}

\usepackage{color}

\newcommand{\blue}[1]{{\color{black} #1}}

\def\defemb#1#2{\expandafter\def\csname #1\endcsname
                              {\relax\ifmmode #2\else\hbox{$#2$}\fi}}
\defemb{cA}{{\cal A}}
\defemb{cB}{{\cal B}}
\defemb{cC}{{\cal C}}
\defemb{cD}{{\cal D}}
\defemb{cE}{{\cal E}}
\defemb{cF}{{\cal F}}
\defemb{cG}{{\cal G}}
\defemb{cH}{{\cal H}}
\defemb{cI}{{\cal I}}
\defemb{cJ}{{\cal J}}
\defemb{cL}{{\cal L}}
\defemb{cM}{{\cal M}}
\defemb{cN}{{\cal N}}
\defemb{cO}{{\cal O}}
\defemb{cP}{{\cal P}}
\defemb{cQ}{{\cal Q}}
\defemb{cR}{{\cal R}}
\defemb{cS}{{\cal S}}
\defemb{cT}{{\cal T}}
\defemb{cU}{{\cal U}}
\defemb{cV}{{\cal V}}
\defemb{cX}{{\cal X}}
\defemb{cZ}{{\cal Z}}

\newcommand{\pos}{\mathsf{{\cP}os}}

\newcommand{\var}{\mathsf{{\cV}ar}}

\newcommand{\dom}{\mathsf{{\cD}om}}

\newcommand{\sleq}{\leqslant}

\newcommand{\nil}{[\,]}

\newcommand{\Var}{\mathsf{{\cal V}ar}} % variables in an object
\renewcommand{\emptyset}{\varnothing}
\renewcommand{\phi}{\varphi}

\newcommand{\toppos}{\epsilon} % top position
\newcommand{\ol}[1]{\overline{#1}}  % sequence of objects
\newcommand{\id}{id}

\newcommand{\inn}{\stackrel{\mathsf{i}}{\rightarrow}}
\newcommand{\inns}{\stackrel{\mathsf{i}}{\rightarrow}\!\!{}^\ast}
\newcommand{\innm}{\stackrel{\mathsf{i}}{\rightarrow}\!\!{}^m}
\newcommand{\cinn}{\stackrel{\mathsf{c}}{\rightarrow}}
\newcommand{\cinns}{\stackrel{\mathsf{c}}{\rightarrow}\!\!{}^\ast}
\newcommand{\topr}{\stackrel{\mathsf{\epsilon}}{\rightarrow}}
\newcommand{\tops}{\stackrel{\mathsf{\epsilon}}{\rightarrow}\!\!{}^\ast}

\newcommand{\rh}{\rightharpoonup}
\newcommand{\lh}{\leftharpoondown}
\newcommand{\rlh}{\rightleftharpoons}

\newcommand{\hs}{\pi}

\newcommand{\range}{\mathsf{range}}

\def\res{\mathrel{\vert\grave{ }}}

\def \tuple#1{\langle #1 \rangle}

\long\def\comment#1{}

\def\cC{{\mathcal{C}}}
\def\cD{{\mathcal{D}}}
\def\cF{{\mathcal{F}}}

\def\cR{{\mathcal{R}}}
\def\cT{{\mathcal{T}}}
\def\cV{{\mathcal{V}}}
\def\tto{\twoheadrightarrow}

\def\Var{{\mathcal{V}ar}}

\newcommand{\pc}{\mbox{pcDCTRS}}
\newcommand{\pcs}{\mbox{pcDCTRSs}}

\def\interval#1#2{{#1,#2}}

\begin{document}

\begin{frontmatter}

\title{Reversible Computation in Term Rewriting\tnoteref{t1}}
\tnotetext[t1]{This work has been partially supported by the EU
  (FEDER) and the Spanish \emph{Ministerio de Econom\'{\i}a y
    Competitividad} (MINECO) under grants TIN2013-44742-C4-1-R and
  TIN2016-76843-C4-1-R, by the \emph{Generalitat Valenciana} under
  grant PROMETEO-II/2015/013 (SmartLogic), and by the COST Action
  IC1405 on Reversible Computation - extending
  horizons of computing.
  Adri\'an Palacios was partially supported by the EU (FEDER) and
  the Spanish \emph{Ayudas para contratos predoctorales para la
    formaci\'on de doctores} and \emph{Ayudas a la movilidad
    predoctoral para la realizaci\'on de estancias breves en centros
    de I+D}, MINECO (SEIDI), under FPI grants BES-2014-069749 and
  EEBB-I-16-11469.
  Part of this research was done while the second and third authors
  were visiting Nagoya University; they gratefully acknowledge their
  hospitality. \\
  \textcopyright ~2017. This manuscript version is made available under the
  CC-BY-NC-ND 4.0 license 
  \texttt{http://creativecommons.org/licenses/by-nc-nd/4.0/}
}
  \author[ngy]{Naoki Nishida}
  \ead{nishida@i.nagoya-u.ac.jp}
  \author[vlc]{Adri\'an Palacios} \ead{apalacios@dsic.upv.es}
  \author[vlc]{Germ\'an Vidal\corref{cor}} \ead{gvidal@dsic.upv.es}
  \address[ngy]{Graduate School of Informatics, Nagoya University\\
    Furo-cho, Chikusa-ku, 4648603 Nagoya, Japan}
  \address[vlc]{MiST, DSIC, Universitat Polit\`ecnica de Val\`encia\\
    Camino de Vera, s/n, 46022 Valencia, Spain}

\cortext[cor]{Corresponding author.}

% Author macros::end %%%%%%%%%%%%%%%%%%%%%%%%%%%%%%%%%%%%%%%%%%%%%%%%%

\begin{abstract}
  Essentially, in a reversible programming language, for each forward
  computation from state $S$ to state $S'$, there exists a
  constructive method to go backwards from state
  $S'$ to state $S$.  Besides its theoretical interest, reversible
  computation is a fundamental concept which is relevant in many
  different areas like cellular automata, bidirectional program
  transformation, or quantum computing, to name a few.
  
  In this work, we focus on term rewriting, a computation model that
  underlies most rule-based programming languages. In general, term
  rewriting is not reversible, even for injective functions; namely,
  given a rewrite step $t_1 \to t_2$, we do not always have a
  decidable method to get $t_1$ from $t_2$.
  Here, we introduce a conservative extension of term rewriting that
  becomes reversible. 
  Furthermore, we also define two transformations, injectivization
  and inversion, to make a rewrite system reversible using standard
  term rewriting. We illustrate the usefulness of our transformations
  in the context of bidirectional program transformation.\\

  \noindent
  \emph{To appear in the Journal of Logical and Algebraic Methods in Programming}.
\end{abstract}

\begin{keyword}
term rewriting \sep reversible computation \sep program transformation
\end{keyword}

\end{frontmatter}

%%%%%%%%%%%%%%%%%%%%%%%%%%%%%%%%%%%%%%%%%%%%%%%%%%%%%%%%%%%%%%%%%%
\section{Introduction}
\label{sec:intro}

The notion of reversible computation can be traced back to Landauer's
pioneering work \cite{Lan61}. Although Landauer was mainly concerned
with the energy consumption of erasing data in irreversible computing,
he also claimed that every computer can be made reversible by saving
the \emph{history} of the computation. However, as Landauer himself
pointed out, this would only postpone the problem of erasing the tape
of a reversible Turing machine before it could be reused.  Bennett
\cite{Ben73} improved the original proposal so that the computation
now ends with a tape that only contains the output of a computation
and the initial source, thus deleting all remaining ``garbage'' data,
though it performs twice the usual computation steps. More recently,
Bennett's result is extended in \cite{CP95} to nondeterministic Turing
machines, where it is also proved that transforming an irreversible
Turing machine into a reversible one can be done with a quadratic loss
of space.
We refer the interested reader to, e.g., \cite{Ben00,Fra05,Yok10} for
a high level account of the principles of reversible computation.

In the last decades, reversible computing and \emph{reversibilization}
(transforming an irreversible computation device into a reversible
one) have been the subject of intense research, giving rise to
successful applications in many different fields, e.g., cellular
automata \cite{Mor12}, where reversibility is an essential property,
bidirectional program transformation \cite{MHNHT07}, where
reversibility helps to automate the generation of inverse functions
(see Section~\ref{sec:applications}), reversible debugging
\cite{GLM14}, where one can go both forward and backward when seeking
the cause of an error, parallel discrete event simulation
\cite{SJBOQ15}, where reversible computation is used to undo the
effects of speculative computations made on a wrong assumption,
quantum computing \cite{Yam14}, where all computations should be
reversible, and so forth. 
The interested reader can find detailed surveys in the \emph{state of
  the art} reports of the different working groups of COST Action
IC1405 on Reversible Computation \cite{COST}. 

In this work, we introduce reversibility in the context of \emph{term
  rewriting} \cite{BN98,Terese03}, a computation model that underlies
most rule-based programming languages.
In contrast to other, more \emph{ad-hoc} approaches, we consider that
term rewriting is an excellent framework to rigorously define
reversible computation in a functional context and formally prove its
main properties. We expect our work to be useful in different
(sequential) contexts, like reversible debugging, parallel discrete
event simulation or bidirectional program transformation, to name a
few. In particular, Section~\ref{sec:applications} presents a first
approach to formalize bidirectional program transformation in our
setting.

To be more precise, we present a general and intuitive notion of
\emph{reversible} term rewriting by defining a Landauer
embedding. Given a rewrite system $\cR$ and its associated (standard)
rewrite relation $\to_\cR$, we define a reversible extension of
rewriting with two components: a forward relation $\mathrel{\rh_\cR}$
and a backward relation $\mathrel{\lh_\cR}$, such that
$\mathrel{\rh_\cR}$ is a conservative extension of $\to_\cR$ and,
moreover, $\mathrel{(\rh_\cR)^{-1}} \,=\, \mathrel{\lh_\cR}$.
We note that the inverse rewrite relation, $\mathrel{(\to_\cR)^{-1}}$,
is not an appropriate basis for ``reversible'' rewriting since we aim
at defining a technique to \emph{undo} a particular reduction.  In
other words, given a rewriting reduction $s \to^\ast_\cR t$, our
reversible relation aims at computing the term $s$ from $t$ and $\cR$
in a decidable and deterministic way, which is not possible using
$\mathrel{(\to_\cR)^{-1}}$ since it is generally non-deterministic,
non-confluent, and non-terminating, even for systems defining
injective functions (see Example~\ref{ex:rel}).  In contrast, our
backward relation $\lh_\cR$ is deterministic (thus confluent) and
terminating.
Moreover, our relation proceeds backwards step by step, i.e., the
number of reduction steps in $s \rh^\ast_\cR t$ and $t \lh^\ast_\cR s$
are the same.

In order to introduce a reversibilization transformation for rewrite
systems, we use a \emph{flattening} transformation so that the
reduction at top positions of terms suffices to get a normal form in
the transformed systems.
For instance, given the following rewrite system:
\[
\begin{array}{lrcll}
  & \mathsf{add}(\mathsf{0},y) & \to & y,~\\
  & \mathsf{add}(\mathsf{s}(x),y) &
  \to & \mathsf{s}(\mathsf{add}(x,y)) &  \\
\end{array}
\]
defining the addition on natural numbers built from constructors
$\mathsf{0}$ and $\mathsf{s}(~)$, we produce the following
\emph{flattened} (conditional) system: 
\[
\begin{array}{lrcll}
  \cR = \{ & \mathsf{add}(\mathsf{0},y) & \to & y,\\
  & \mathsf{add}(\mathsf{s}(x),y) & \to & \mathsf{s}(z) \Leftarrow
  \mathsf{add}(x,y)\tto z & \}\\
\end{array}
\]
(see Example~\ref{ex:basic} for more details).
This allows us to provide an improved notion of reversible rewriting
in which some information (namely, the positions where reduction
takes place) is not required anymore. This opens the door to
\emph{compile} the reversible extension of rewriting into the system
rules. Loosely speaking, given a system $\cR$, we produce new systems
$\cR_f$ and $\cR_b$ such that \emph{standard} rewriting in $\cR_f$,
i.e., $\to_{\cR_f}$, coincides with the forward reversible extension
$\rh_\cR$ in the original system, and analogously $\to_{\cR_b}$ is
equivalent to $\lh_\cR$. E.g., for the system $\cR$ above, we would
produce
\[
\begin{array}{l@{~}r@{~}c@{~}l@{~}l}
\cR_f = \{ & \mathsf{add}^\mathtt{i}(\mathsf{0},y) & \to &
\tuple{y,\beta_1},\\
& \mathsf{add}^\mathtt{i} (\mathsf{s}(x),y) & \to &
\tuple{\mathsf{s}(z),\beta_2(w)} \Leftarrow \mathsf{add}^\mathtt{i} (x,y)\tto
\tuple{z,w} & \} \\[2ex]
  \cR_b = \{ &\mathsf{add}^{-1}(y,\beta_1)& \to&
                 \tuple{\mathsf{0},y},\\
  & \mathsf{add}^{-1}(\mathsf{s}(z),\beta_2(w)) & \to &
  \tuple{\mathsf{s}(x),y} \Leftarrow \mathsf{add}^{-1}(z,w)\to
  \tuple{x,y} & \}
\end{array}
\]
where $\mathsf{add}^\mathtt{i}$ is an injective version of function
$\mathsf{add}$, $\mathsf{add}^{-1}$ is the inverse of
$\mathsf{add}^\mathtt{i}$, and $\beta_1,\beta_2$ are fresh symbols
introduced to label the rules of $\cR$.

In this work, we will mostly consider \emph{conditional} rewrite
systems, not only to have a more general notion of reversible
rewriting, but also to define a reversibilization technique for
unconditional rewrite systems, since the application of
\emph{flattening} (cf.\ Section~\ref{sec:irr}) may introduce
conditions in a system that is originally unconditional, as
illustrated above.

This paper is an extended version of \cite{NPV16}. In contrast to
\cite{NPV16}, our current paper includes the proofs of technical
results, the reversible extension of term rewriting is introduced
first in the unconditional case (which is simpler and more intuitive),
and presents an improved injectivization transformation when the
system includes injective functions. Furthermore, a prototype
implementation of the reversibilization technique is publicly
available from \texttt{http://kaz.dsic.upv.es/rev-rewriting.html}.
 
The paper is organized as follows. After introducing some
preliminaries in Section~\ref{sec:prelim}, we present our approach to
reversible term rewriting in Section~\ref{sec:rr}.
Section~\ref{sec:irr} introduces the class of \emph{pure constructor}
systems where all reductions take place at topmost positions, so that
storing this information in reversible rewrite steps becomes
unnecessary. Then, Section~\ref{sec:transf} presents injectivization
and inversion transformations in order to make a rewrite system
reversible with standard rewriting. Here, we also present an
improvement of the transformation for injective functions. The
usefulness of these transformations is illustrated in
Section~\ref{sec:applications}.  Finally, Section~\ref{sec:relwork}
discusses some related work and Section~\ref{sec:conclusion} concludes
and points out some ideas for future research.

%%%%%%%%%%%%%%%%%%%%%%%%%%%%%%%%%%%%%%%%%%%%%%%%%%%%%%%%%%%%%%%%%%
\section{Preliminaries}
\label{sec:prelim}

We assume familiarity with basic concepts of term rewriting. We refer
the reader to, e.g., \cite{BN98} and \cite{Terese03} %\cite{ohlebusch}
for further details.

\subsection{Terms and Substitutions}

A \emph{signature} $\cF$ is a set of ranked function symbols.  Given a set of
variables $\cV$ with $\cF\cap\cV=\emptyset$, we denote the domain of
\emph{terms} by $\cT(\cF,\cV)$.  We use $\mathsf{f},\mathsf{g},\ldots$
to denote functions and $x,y,\ldots$ to denote variables.
Positions are used to address the nodes of a term viewed as a tree. A
\emph{position} $p$ in a term $t$, in symbols $p\in\pos(t)$, is
represented by a finite sequence of natural numbers, where $\toppos$
denotes the root position.
We let $t|_p$ denote the \emph{subterm} of $t$ at position $p$ and
$t[s]_p$ the result of \emph{replacing the subterm} $t|_p$ by the term
$s$.
$\Var(t)$ denotes the set of variables appearing in $t$.  We also let
$\Var(t_1,\ldots,t_n)$ denote $\Var(t_1)\cup\cdots\cup\Var(t_n)$.  A
term $t$ is \emph{ground} if $\Var(t) = \emptyset$.

A \emph{substitution} $\sigma : \cV \mapsto \cT(\cF,\cV)$ is a mapping
from variables to terms such that $\dom(\sigma) = \{x \in \cV \mid x
\neq \sigma(x)\}$ is its domain.  
A substitution $\sigma$ is \emph{ground} if $x\sigma$ is ground for
all $x\in\dom(\sigma)$.
Substitutions are extended to morphisms from $\cT(\cF,\cV)$ to
$\cT(\cF,\cV)$ in the natural way.  We denote the application of a
substitution $\sigma$ to a term $t$ by $t\sigma$ rather than
$\sigma(t)$. The identity substitution is denoted by $\id$.
We let ``$\circ$'' denote the composition of substitutions, i.e.,
$\sigma\circ\theta(x) = (x\theta)\sigma = x\theta\sigma$.
The \emph{restriction} $\theta\!\res_V$ of a substitution $\theta$ to
a set of variables $V$ is defined as follows: $x\theta\!\res_{V} =
x\theta$ if $x\in V$ and $x\theta\!\res_V = x$ otherwise. 

\subsection{Term Rewriting Systems}

A set of rewrite rules $l \to r$ such that $l$ is a nonvariable term
and $r$ is a term whose variables appear in $l$ is called a \emph{term
  rewriting system} (TRS for short); terms $l$ and $r$ are called the
left-hand side and the right-hand side of the rule, respectively. We
restrict ourselves to finite signatures and TRSs.  Given a TRS $\cR$
over a signature $\cF$, the \emph{defined} symbols $\cD_\cR$ are the
root symbols of the left-hand sides of the rules and the
\emph{constructors} are $\cC_\cR = \cF \;\backslash\; \cD_\cR$.
\emph{Constructor terms} of $\cR$ are terms over $\cC_\cR$ and $\cV$,
denoted by $\cT(\cC_\cR,\cV)$.  We sometimes omit $\cR$ from $\cD_\cR$
and $\cC_\cR$ if it is clear from the context.  A substitution
$\sigma$ is a \emph{constructor substitution} (of $\cR$) if $x\sigma
\in \cT(\cC_\cR,\cV)$ for all variables $x$.
 
For a TRS $\cR$, we define the associated rewrite relation $\to_\cR$
as the smallest binary relation on terms satisfying the following: given terms
$s,t\in\cT(\cF,\cV)$, we have $s \to_\cR t$ iff there exist a position
$p$ in $s$, a rewrite rule $l \to r \in \cR$, and a substitution
$\sigma$ such that $s|_p = l\sigma$ and $t = s[r\sigma]_p$; the rewrite
step is sometimes denoted by $s \to_{p,l\to r} t$ to make explicit the
position and rule used in this step.  
The instantiated left-hand side $l\sigma$ is called a \emph{redex}.
A term $s$ is called \emph{irreducible} or in \emph{normal form}
with respect to a TRS $\cR$ if there is no term $t$ with $s \to_\cR t$.  
A substitution is called \emph{normalized} with respect to $\cR$ if every
variable in the domain is replaced by a normal form with respect to $\cR$.  We
sometimes omit ``with respect to $\cR$'' if it is clear from the context.  
A \emph{derivation} is a (possibly empty) sequence of rewrite steps.
Given a binary relation $\to$, we denote by $\rightarrow^{\ast}$ its
reflexive and transitive closure, i.e., $s \to^\ast_\cR t$ means that
$s$ can be reduced to $t$ in $\cR$ in zero or more steps; we also use
$s \to^n_\cR t$ to denote that $s$ can be reduced to $t$ in exactly
$n$ steps.

We further assume that rewrite rules are labelled, i.e., given a TRS
$\cR$, we denote by $\beta:l\to r$ a rewrite rule with label
$\beta$. Labels are unique in a TRS.  Also, to relate label $\beta$ to
fixed variables, we consider that the variables of the rewrite rules
are not renamed\footnote{This will become useful in the next section
  where the reversible extension of rewriting keeps a ``history'' of a
  computation in the form of a list of terms $\beta(p,\sigma)$, and we
  want the domain of $\sigma$ to be a subset of the left-hand side of
  the rule labelled with $\beta$. } and that the reduced terms are
always ground. Equivalently, one could require terms to be variable
disjoint with the variables of the rewrite system, but we require
groundness for simplicity.
We often write $s\to_{p,\beta} t$ instead of $s \to_{p,l\to r} t$ if
rule $l\to r$ is labeled with $\beta$.

\subsection{Conditional Term Rewrite Systems}

In this paper, we also consider \emph{conditional} term rewrite
systems (CTRSs); namely oriented 3-CTRSs, i.e., CTRSs where extra
variables are allowed as long as $\Var(r)\subseteq\Var(l)\cup\Var(C)$
for any rule $l\to r\Leftarrow C$ \cite{MH94}.
In \emph{oriented} CTRSs, a conditional rule $l\to r \Leftarrow C$ has the
form $l\to r \Leftarrow s_1\tto t_1,\ldots,s_n\tto t_n$, where each
oriented equation $s_i\tto t_i$ is interpreted as reachability
($\to_\cR^\ast$). 
In the following, we denote by $\ol{o_n}$ a sequence of elements
$o_1,\ldots,o_n$ for some $n$. We also write $\ol{o_\interval{i}{j}}$
for the sequence $o_i,\ldots,o_j$ when $i\leq j$ (and the empty
sequence otherwise).  We write $\ol{o}$ when the number of
elements is not relevant.  In addition, we denote a condition
$o_1\tto o'_1,\ldots,o_n\tto o'_n$ by $\ol{o_n\tto o'_n}$.

As in the unconditional case, we consider that rules are labelled and
that labels are unique in a CTRS. And, again, to relate label $\beta$
to fixed variables, we consider that the variables of the conditional
rewrite rules are not renamed and that the reduced terms are always
ground.

For a CTRS $\cR$, the associated rewrite relation $\to_\cR$ is defined
as the smallest binary relation satisfying the following: given
ground terms $s,t\in\cT(\cF)$, we have $s \to_\cR t$ iff there exist
a position $p$ in $s$, a rewrite rule $l \to r\Leftarrow \ol{s_n\tto
  t_n} \in \cR$, and a ground substitution $\sigma$ such that $s|_p
= l\sigma$, $s_i\sigma \to^\ast_\cR t_i\sigma$ for all $i=1,\ldots,n$,
and $t = s[r\sigma]_p$.

In order to simplify the presentation, we only consider
\emph{deterministic} CTRSs (DCTRSs), i.e., oriented 3-CTRSs where, for
each rule $l\to r \Leftarrow \ol{s_n\tto t_n}$, we have
$\var(s_i)\subseteq\var(l,\ol{t_{i-1}})$ for all $i=1,\ldots,n$
(see Section~\ref{sec:ctrs} for a justification of this requirement
and how it could be relaxed to arbitrary 3-CTRSs).
Intuitively speaking, the use of DCTRs allows us to compute the
bindings for the variables in the condition of a rule in a
deterministic way. E.g., given a ground term $s$ and a rule $\beta:
l\to r \Leftarrow \ol{s_n\tto t_n}$ with $s|_p = l\theta$, we have
that $s_1\theta$ is ground. Therefore, one can reduce $s_1\theta$ to
some term $s'_1$ such that $s'_1$ is an instance of $t_1\theta$ with
some ground substitution $\theta_1$. Now, we have that
$s_2\theta\theta_1$ is ground and we can reduce $s_2\theta\theta_1$ to
some term $s'_2$ such that $s'_2$ is an instance of
$t_2\theta\theta_1$ with some ground substitution $\theta_2$, and so
forth.  If all equations in the condition hold using
$\theta_1,\ldots,\theta_n$, we have that $s \to_{p,\beta}
s[r\sigma]_p$ with $\sigma=\theta\theta_1\ldots\theta_n$.

\begin{example} \label{ex:dctrs} 
  Consider the following DCTRS $\cR$ that defines the function
  $\mathsf{double}$ that doubles the value of its argument when it is
  an even natural number:
  \[
  \begin{array}{l@{~}r@{~}c@{~}l@{}l@{~}r@{~}c@{~}l}
   \beta_1: & \mathsf{add}(\mathsf{0},y) & \to & y  &
    \beta_4: & \mathsf{even}(\mathsf{0}) & \to & \mathsf{true}\\
    \beta_2: &\mathsf{add}(\mathsf{s}(x),y) & \to &
    \mathsf{s}(\mathsf{add}(x,y)) &
    \beta_5: & \mathsf{even}(\mathsf{s}(\mathsf{s}(x))) & \to & \mathsf{even}(x)\\
    \beta_3: & \mathsf{double}(x) & \to & \mathsf{add}(x,x)
    \Leftarrow \mathsf{even}(x)\tto\mathsf{true}\\
  \end{array}
  \]
  Given the term
  $\mathsf{double}(\mathsf{s}(\mathsf{s}(\mathsf{0})))$ we have, for
  instance, the following derivation:
  \[
  \begin{array}{l@{}l@{}lll}
    \mathsf{double}(\mathsf{s}(\mathsf{s}(\mathsf{0}))) &
    \to_{\toppos,\beta_3} &
    \mathsf{add}(\mathsf{s}(\mathsf{s}(\mathsf{0})),\mathsf{s}(\mathsf{s}(\mathsf{0})))
    & \mbox{since}~\mathsf{even}(\mathsf{s}(\mathsf{s}(\mathsf{0})))\to_\cR^\ast
    \mathsf{true} \\
    &&& \mbox{with}~\sigma =
    \{x\mapsto\mathsf{s}(\mathsf{s}(\mathsf{0}))\} \\
    & \to_{\toppos,\beta_2} &
    \mathsf{s}(\mathsf{add}(\mathsf{s}(\mathsf{0}),\mathsf{s}(\mathsf{s}(\mathsf{0}))))
    & \mbox{with}~\sigma =     \{x\mapsto\mathsf{s}(\mathsf{0}),~y\mapsto\mathsf{s}(\mathsf{s}(\mathsf{0}))\}\\
    & \to_{1,\beta_2} &
    \mathsf{s}(\mathsf{s}(\mathsf{add}(\mathsf{0},\mathsf{s}(\mathsf{s}(\mathsf{0})))))
    & \mbox{with}~
    \sigma=\{x\mapsto\mathsf{0},y\mapsto\mathsf{s}(\mathsf{s}(\mathsf{0}))\}
    \\
    & \to_{1.1,\beta_1} &
    \mathsf{s}(\mathsf{s}(\mathsf{s}(\mathsf{s}(\mathsf{0})))) &
    \mbox{with}~ \sigma = \{y\mapsto \mathsf{s}(\mathsf{s}(\mathsf{0}))\}
  \end{array}
  \]
\end{example}

%%%%%%%%%%%%%%%%%%%%%%%%%%%%%%%%%%%%%%%%%%%%%%%%%%%%%%%%%%%%%%%%%%
\section{Reversible Term Rewriting} \label{sec:rr}

In this section, we present a conservative extension of the rewrite
relation which becomes reversible. 
In the following, we use $\rh_\cR$ to denote our \emph{reversible}
(forward) term rewrite relation, and $\lh_\cR$ to denote its
application in the reverse (backward) direction. Note that, in
principle, we do not require $\lh_\cR \:=\: \rh_\cR^{-1}$, i.e., we
provide independent (constructive) definitions for each
relation. Nonetheless, we will prove that $\lh_\cR \:=\: \rh_\cR^{-1}$
indeed holds (cf.\
Theorems~\ref{th:unconditional}~and~\ref{th:conditional}).
In some approaches to reversible computing, both forward and backward
relations should be deterministic. Here, we will only require
deterministic \emph{backward} steps, while forward steps might be
non-deterministic, as it is often the case in term rewriting.

%%%%%%%%%%%%%%%%%%%%%%%%%%%%%%%%%%%%%%%%%%%%%%%%%%%%%%%%%%%%%%%%%%
\subsection{Unconditional Term Rewrite Systems} \label{sec:trs}

We start with unconditional TRSs since it is conceptually simpler and
thus will help the reader to better understand the key ingredients of
our approach. In the next section, we will consider the more general
case of DCTRSs.

Given a TRS $\cR$, reversible rewriting is defined on pairs
$\tuple{t,\hs}$, where $t$ is a ground term and $\hs$ is a trace (the
``history'' of the computation so far). Here, a \emph{trace} in $\cR$
is a list of \emph{trace terms} of the form $\beta(p,\sigma)$ such
that $\beta$ is a label for some rule $l\to r\in\cR$, $p$ is a
position, and $\sigma$ is a substitution with
$\dom(\sigma)=\var(l)\backslash\var(r)$ which will record the bindings
of erased variables when $\var(l)\backslash\var(r)\neq\emptyset$ (and
$\sigma=\id$ if $\var(l)\backslash\var(r)=\emptyset$).\footnote{Note
  that if a rule $l\to r$ is non-erasing, i.e., $\var(l) = \var(r)$,
  then $\sigma=\id$.} Our trace terms have some similarities with
\emph{proof terms} \cite{Terese03}. However, proof terms do not store
the bindings of erased variables (and, to the best of our knowledge,
they are only defined for unconditional TRSs, while we use trace terms
both for unconditional and conditional TRSs).

Our reversible term rewriting relation is only defined on \emph{safe}
pairs:

\begin{definition} %%[safe pair]
  Let $\cR$ be a TRS. The pair $\tuple{s,\hs}$ is \emph{safe} in $\cR$
  iff, for all $\beta(p,\sigma)$ in $\hs$, $\sigma$ is a ground
  substitution with $\dom(\sigma)=\var(l)\backslash\var(r)$ and
  $\beta:l\to r\in\cR$.
\end{definition}
In the following, we often omit $\cR$ when referring to traces and
safe pairs if the underlying TRS is clear from the context.

Safety is not necessary when applying a forward reduction step, but
will become essential for the backward relation $\lh_\cR$ to be
correct.  E.g., all traces that come from the forward reduction of
some initial pair with an empty trace will be safe (see below).
Reversible rewriting is then introduced as follows:

\begin{definition} 
  Let $\cR$ be a TRS. A reversible rewrite relation $\rh_\cR$ is
  defined on safe pairs $\tuple{t,\hs}$, where $t$ is a ground term and
  $\hs$ is a trace in $\cR$. The reversible rewrite relation extends
  standard rewriting as follows:\footnote{In the following, we
    consider the usual infix notation for lists where $\nil$ is the
    empty list and $x:xs$ is a list with head $x$ and tail $xs$.}
  \[
  \tuple{s,\hs} \rh_\cR \tuple{t,\beta(p,\sigma'):\hs}
  \]
  iff there exist a position $p\in\pos(s)$, a rewrite rule $\beta: l
  \to r \in \cR$, and a ground substitution $\sigma$ such that $s|_p =
  l\sigma$, $t = s[r\sigma]_p$, and $\sigma' =
  \sigma\!\!\res_{\var(l)\backslash\var(r)}$.
  The reverse relation, $\lh_\cR$, is then defined as follows:
  \[
  \tuple{t,\beta(p,\sigma'):\hs} \lh_\cR \tuple{s,\hs} 
  \]
  iff $\tuple{t,\beta(p,\sigma'):\hs}$ is a safe pair in $\cR$ and
  there exist a ground substitution $\theta$ and a rule $\beta:l\to
  r\in\cR$ such that $\dom(\theta)=\var(r)$, $t|_p = r\theta$ and $s =
  t[l\theta\sigma']_p$. Note that $\theta\sigma' = \sigma'\theta =
  \theta\cup\sigma'$, where $\cup$ is the union of substitutions,
  since $\dom(\theta)=\var(r)$,
  $\dom(\sigma')=(\var(l)\backslash\var(r))$ and both substitutions
  are ground, so $\dom(\theta)\cap\dom(\sigma')=\emptyset$.
\end{definition}
We denote the union of both relations $\rh_\cR\cup\lh_\cR$ by
$\rlh_\cR$.

\begin{example} \label{ex:add} Let us consider the following TRS $\cR$
  defining the addition on natural numbers built from $\mathsf{0}$ and
  $\mathsf{s}(~)$, and the function $\mathsf{fst}$ that returns its
  first argument:
  \[
  \begin{array}{lrcl@{~~~~~~~~~~}lrcl}
   \beta_1: & \mathsf{add}(\mathsf{0},y) & \to & y  &
   \beta_3: &\mathsf{fst}(x,y) & \to & x \\
    \beta_2: &\mathsf{add}(\mathsf{s}(x),y) & \to &
    \mathsf{s}(\mathsf{add}(x,y))\\
  \end{array}
  \]
  Given the term $\mathsf{fst}(
  \mathsf{add}(\mathsf{s}(\mathsf{0}),\mathsf{0}), \mathsf{0})$, we
  have, for instance, the following reversible (forward) derivation:
  \[
  \begin{array}{lll}
  \tuple{\mathsf{fst}( \mathsf{add}(\mathsf{s}(\mathsf{0}),\mathsf{0}), \mathsf{0}),\blue{\nil}}
  & \rh_\cR & 
  \tuple{\mathsf{fst}(
    \mathsf{s}(\mathsf{add}(\mathsf{0},\mathsf{0})), \mathsf{0}),\blue{[\beta_2(1,\id)]}} \\
  & \rh_\cR & 
  \tuple{\mathsf{s}(\mathsf{add}(\mathsf{0},\mathsf{0})),\blue{[\beta_3(\epsilon,\{y\mapsto\mathsf{0}\}),\beta_2(1,\id)]}} \\
  & \rh_\cR & 
  \tuple{\mathsf{s}(\mathsf{0}),\blue{[\beta_1(1,\id),
      \beta_3(\epsilon,\{y\mapsto\mathsf{0}\}),\beta_2(1,\id)]}} \\
  \end{array}
  \]
  The reader can easily check that
  $\tuple{\mathsf{s}(\mathsf{0}),\blue{[\beta_1(1,\id),
      \beta_3(\epsilon,\{y\mapsto\mathsf{0}\}),\beta_2(1,\id)]}}$ is
  reducible to $\tuple{\mathsf{fst}(
    \mathsf{add}(\mathsf{s}(\mathsf{0}),\mathsf{0}),
    \mathsf{0}),\blue{\nil}}$ using the backward relation $\lh_\cR$ by
  performing exactly the same steps but in the backward direction.
\end{example}
An easy but essential property of $\rh_\cR$ is that it is a
conservative extension of standard rewriting in the following sense
(we omit its proof since it is straightforward):

\begin{theorem}
  Let $\cR$ be a TRS. Given terms $s,t$, if $s\to_\cR^\ast t$, then
  for any trace $\hs$ there exists a trace $\hs'$ such that 
  $\tuple{s,\hs} \rh_\cR^\ast \tuple{t,\hs'}$.
\end{theorem}
Here, and in the following, we assume that $\mathrel{\leftarrow_\cR}
\,=\, \mathrel{(\to_\cR)^{-1}}$, i.e., $s \to_\cR^{-1} t$ is denoted
by $s \mathrel{\leftarrow_\cR} t$.
Observe that the backward relation is not a conservative extension of
$\mathrel{\leftarrow_\cR}$: in general, $t \leftarrow_\cR s$ does not
imply $\tuple{t,\hs'} \lh_\cR \tuple{s,\hs}$ for any arbitrary trace
$\hs'$.
This is actually the purpose of our notion
of reversible rewriting: $\lh_\cR$ should not extend $\leftarrow_\cR$
but is only aimed at performing \emph{exactly the same steps} of the
forward computation whose trace was stored, but in the reverse order.
Nevertheless, one can still ensure that for all steps
$t \leftarrow_\cR s$, there exists some trace $\pi'$ such that
$\tuple{t,\hs'} \lh_\cR \tuple{s,\hs}$ (which is an easy consequence of the
above result and Theorem~\ref{th:unconditional} below).

\begin{example} \label{ex:rel} Consider again the following TRS $\cR =
  \{ \beta:~\mathsf{snd}(x,y) \to y\}$. Given the reduction
  $\mathsf{snd}(\mathsf{1},\mathsf{2}) \to_\cR \mathsf{2}$, there are
  infinitely many reductions for $\mathsf{\mathsf{2}}$ using
  $\leftarrow_\cR$, e.g., $\mathsf{2} \leftarrow_\cR
  \mathsf{snd}(\mathsf{1},\mathsf{2})$, $\mathsf{2} \leftarrow_\cR
  \mathsf{snd}(\mathsf{2},\mathsf{2})$, $\mathsf{2} \leftarrow_\cR
  \mathsf{snd}(\mathsf{3},\mathsf{2})$, etc. The relation is also
  non-terminating: $\mathsf{2} \leftarrow_\cR
  \mathsf{snd}(\mathsf{1},\mathsf{2}) \leftarrow_\cR
  \mathsf{snd}(\mathsf{1},\mathsf{snd}(\mathsf{1},\mathsf{2}))
  \leftarrow_\cR \cdots$. In contrast, given a pair
  $\tuple{\mathsf{2},\hs}$, we can only perform a single deterministic
  and finite reduction (as proved below). For instance, if
  $\hs=[\beta(\epsilon,\{x\mapsto\mathsf{\mathsf{1}}\}),\beta({2},\{x\mapsto\mathsf{\mathsf{1}}\})]$,
  then the only possible reduction is $\tuple{\mathsf{\mathsf{2}},\hs}
  \lh_\cR
  \tuple{\mathsf{snd}(\mathsf{\mathsf{1}},\mathsf{\mathsf{2}}),[\beta({2},\{x\mapsto\mathsf{\mathsf{1}}\})]}
  \lh_\cR
  \tuple{\mathsf{snd}(\mathsf{\mathsf{1}},\mathsf{snd}(\mathsf{\mathsf{1}},\mathsf{\mathsf{2}})),\nil}
  \not\lh_\cR$.
\end{example}
Now, we state a lemma which shows that safe pairs are preserved
through reversible term rewriting (both in the forward and backward
directions):

\begin{lemma} \label{lemma:safe} Let $\cR$ be a TRS. Let
  $\tuple{s,\hs}$ be a safe pair. If $\tuple{s,\hs} \rlh^\ast_\cR
  \tuple{t,\hs'}$, then $\tuple{t,\hs'}$ is also safe.
\end{lemma}

\begin{proof}
  We prove the claim by induction on the length $k$ of the
  derivation. Since the base case $k=0$ is trivial, consider the
  inductive case $k>0$. Assume a derivation of the form $\tuple{s,\hs}
  \rlh_\cR^\ast \tuple{s_0,\hs_0} \rlh_\cR \tuple{t,\hs'}$. By the
  induction hypothesis, we have that $\tuple{s_0,\hs_0}$ is a safe
  pair. Now, we distinguish two cases depending on the last step. If
  we have $\tuple{s_0,\hs_0} \rh_\cR \tuple{t,\hs'}$, then there exist
  a position $p\in\pos(s_0)$, a rewrite rule $\beta: l \to r \in \cR$,
  and a ground substitution $\sigma$ such that $s_0|_p = l\sigma$, $t
  = s_0[r\sigma]_p$, $\sigma' =
  \sigma\!\!\res_{\var(l)\backslash\var(r)}$, and
  $\hs'=\beta(p,\sigma'):\hs_0$. Then, since $\sigma'$ is ground and
  $\dom(\sigma')=\var(l)\backslash\var(r)$ by construction, the claim
  follows straightforwardly.
  If the last step has the form $\tuple{s_0,\hs_0} \lh_\cR
  \tuple{t,\hs'}$, then the claim follows trivially since
  each step with $\lh_\cR$ only removes trace terms from $\hs_0$.
  \qed
\end{proof}
Hence, since any pair with an empty trace is safe the following
result, which states that every pair that is \emph{reachable} from
an initial pair with an empty trace is safe, straightforwardly
follows from Lemma~\ref{lemma:safe}:

\begin{proposition} \label{th:safe}
  Let $\cR$ be a TRS. If $\tuple{s,\nil} \rlh_\cR^\ast \tuple{t,\pi}$,
  then $\tuple{t,\pi}$ is safe.
\end{proposition}
Now, we state the reversibility of $\rh_\cR$, i.e., the fact that
$(\rh_\cR)^{-1} = \:\lh_\cR$ (and thus the reversibility of
$\lh_\cR$ and $\rlh_\cR$, too).

\begin{theorem} \label{th:unconditional} Let $\cR$ be a TRS. Given the
  safe pairs $\tuple{s,\hs}$ and $\tuple{t,\hs'}$, for all $n\geq 0$,
  $\tuple{s,\hs} \rh_\cR^n \tuple{t,\hs'}$ iff $\tuple{t,\hs'}
  \lh_\cR^n \tuple{s,\hs}$. 
\end{theorem}

\begin{proof}
  $(\Rightarrow)$ We prove the claim by induction on the length $n$ of
  the derivation $\tuple{s,\hs} \rh_\cR^n \tuple{t,\hs'}$. Since the
  base case $n=0$ is trivial, let us consider the inductive case
  $n>0$. Consider a derivation $\tuple{s,\hs} \rh_\cR^{n-1}
  \tuple{s_0,\hs_0} \rh_\cR \tuple{t,\hs'}$.  By
  Lemma~\ref{lemma:safe}, both $\tuple{s_0,\hs_0}$ and
  $\tuple{t,\hs'}$ are safe.  By the induction hypothesis, we have
  $\tuple{s_0,\hs_0} \lh_\cR^{n-1} \tuple{s,\hs}$. Consider now the
  step $\tuple{s_0,\hs_0} \rh_\cR \tuple{t,\hs'}$. Then, there is a
  position $p\in\pos(s_0)$, a rule $\beta:l \to r\in\cR$ and a ground
  substitution $\sigma$ such that $s_0|_p = l\sigma$,
  $t=s_0[r\sigma]_p$, $\sigma' =
  \sigma\!\!\res_{\var(l)\backslash\var(r)}$, and
  $\hs'=\beta(p,\sigma'):\hs_0$. Let $\theta =
  \sigma\!\!\res_{\var(r)}$. Then, we have $\tuple{t,\hs'} \lh_\cR
  \tuple{s'_0,\pi_0}$ with $t|_p=r\theta$, $\beta:l\to r\in\cR$ and
  $s'_0=t[l\theta\sigma']_p$. Moreover, since $\sigma =
  \theta\sigma'$, we have $s'_0 = t[l\theta\sigma']_p =
  t[l\sigma]_p = s_0$, and the claim follows.

  $(\Leftarrow)$ This direction proceeds in a similar way. We prove
  the claim by induction on the length $n$ of the derivation
  $\tuple{t,\hs'} \lh_\cR^n \tuple{s,\hs}$. As before, we only
  consider the inductive case $n>0$. Let us consider a derivation
  $\tuple{t,\hs'} \lh_\cR^{n-1} \tuple{s_0,\hs_0} \lh_\cR
  \tuple{s,\hs}$.  By Lemma~\ref{lemma:safe}, both $\tuple{s_0,\hs_0}$
  and $\tuple{s,\hs}$ are safe. By the induction hypothesis, we have
  $\tuple{s_0,\hs_0} \rh_\cR^{n-1} \tuple{t,\hs'}$.  Consider now the
  reduction step $\tuple{s_0,\hs_0} \lh_\cR \tuple{s,\hs}$. Then, we
  have $\hs_0 = \beta(p,\sigma'):\hs$, $\beta:l \to r\in\cR$, and
  there exists a ground substitution $\theta$ with
  $\dom(\theta)=\var(r)$ such that $s_0|_p= r\theta$ and
  $s=s_0[l\theta\sigma']_p$. Moreover, since $\tuple{s_0,\hs_0}$ is
  safe, we have that $\dom(\sigma')=\var(l)\backslash\var(r)$ and,
  thus, $\dom(\theta)\cap\dom(\sigma')=\emptyset$. Let $\sigma =
  \theta\sigma'$. Then, since $s|_p = l\sigma$ and
  $\dom(\sigma')=\var(l)\backslash\var(r)$, we can perform the step
  $\tuple{s,\hs} \rh_\cR \tuple{s'_0,\beta(p,\sigma'):\hs}$ with $s'_0
  = s[r\sigma]_p = s[r\theta\sigma']_p = s[r\theta]_p = s_0[r\theta]_p
  = s_0$, and the claim follows. \qed
\end{proof}
The next corollary is then immediate:

\begin{corollary} \label{th:rev2} Let $\cR$ be a TRS. Given the safe
  pairs $\tuple{s,\hs}$ and $\tuple{t,\hs'}$, for all $n\geq 0$, $\tuple{s,\hs}
  \rlh_\cR^n \tuple{t,\hs'}$ iff $\tuple{t,\hs'} \rlh_\cR^n
  \tuple{s,\hs}$.
\end{corollary}
A key issue of our notion of reversible rewriting is that the backward
rewrite relation $\lh_\cR$ is deterministic (thus confluent),
terminating, and has a constructive definition:

\begin{theorem} \label{th:deterministic}
  Let $\cR$ be a TRS. Given a safe pair $\tuple{t,\hs'}$, there exists
  at most one pair $\tuple{s,\hs}$ such that $\tuple{t,\hs'} \lh_\cR
  \tuple{s,\hs}$.
\end{theorem}

\begin{proof}
  First, if there is no step using $\lh_\cR$ from $\tuple{t,\hs'}$,
  the claim follows trivially. Now, assume there is at least one step
  $\tuple{t,\hs'} \lh_\cR \tuple{s,\hs}$. We prove that this is the
  only possible step. By definition, we have
  $\hs'=\beta(p,\sigma'):\hs$, $p\in\pos(t)$, $\beta:l\to r\in\cR$,
  and there exists a ground substitution $\theta$ with
  $\dom(\theta)=\var(r)$ such that $t|_p=r\theta$ and
  $s=t[l\theta\sigma']_p$. The only source of nondeterminism may come
  from choosing a rule labeled with $\beta$ and from the computation
  of the substitution $\theta$.  The claim follows trivially from the
  fact that labels are unique in $\cR$ and that, if there is some
  ground substitution $\theta'$ with $\theta'=\var(r)$ and
  $t|_p=r\theta'$, then $\theta=\theta'$. \qed
\end{proof} 
Therefore, $\lh_\cR$ is clearly deterministic and confluent.
Termination holds straightforwardly for pairs with finite traces since
its length strictly decreases with every backward step. Note however
that even when $\rh_\cR$ and $\lh_\cR$ are terminating, the relation
$\rlh_\cR$ is always non-terminating since one can keep going back and
forth.

%%%%%%%%%%%%%%%%%%%%%%%%%%%%%%%%%%%%%%%%%%%%%%%%%%%%%%%%%%%%%%%%%%
\subsection{Conditional Term Rewrite Systems} \label{sec:ctrs}

In this section, we extend the previous notions and results to DCTRSs.
We note that considering DCTRSs is not enough to make conditional
rewriting deterministic. In general, given a rewrite step $s
\to_{p,\beta} t$ with $p$ a position of $s$, $\beta: l\to r \Leftarrow
\ol{s_n\to t_n}$ a rule, and $\sigma$ a substitution such that $s|_p =
l\sigma$ and $s_i\sigma \to^\ast_\cR t_i\sigma$ for all
$i=1,\ldots,n$, there are three potential sources of non-determinism:
the selected position $p$, the selected rule $\beta$, and the
substitution $\sigma$. The use of DCTRSs can only make deterministic
the last one, but the choice of a position and the selection of a rule
may still be non-deterministic.

For DCTRSs, the notion of a trace term used for TRSs is not sufficient
since we also need to store the traces of the subderivations
associated to the condition of the applied rule (if any). Therefore,
we generalize the notion of a trace as follows:

\begin{definition}[trace]
  Given a CTRS $\cR$, a \emph{trace} in $\cR$ is recursively defined
  as follows:
  \begin{itemize}
  \item the empty list is a trace;
  \item if $\hs,\hs_1,\ldots,\hs_n$ are traces in $\cR$, $n\geq 0$,
    $\beta:l\to r \Leftarrow \ol{s_n\tto t_n}\in\cR$ is a rule,
    $p$ is a position, and $\sigma$ is a ground substitution, then
    $\beta(p,\sigma,\hs_1,\ldots,\hs_n):\hs$ is a trace in $\cR$.
  \end{itemize}
  We refer to each component $\beta(p,\sigma,\hs_1,\ldots,\hs_n)$ in a
  trace as a \emph{trace term}.
\end{definition}
Intuitively speaking, a trace term
$\beta(p,\sigma,\hs_1,\ldots,\hs_n)$ stores the position of a
reduction step, a substitution with the bindings that are required for
the step to be reversible (e.g., the bindings for the erased
variables, but not only; see below) and the traces associated to the
subcomputations in the condition.

The notion of a safe pair is now more involved in order to deal with
conditional rules. The motivation for this definition will be
explained below, after introducing reversible rewriting for DCTRSs.

\begin{definition}[safe pair] 
  Let $\cR$ be a DCTRS. A trace $\hs$ is \emph{safe} in $\cR$ iff, for
  all trace terms $\beta(p,\sigma,\ol{\hs_n})$ in $\hs$, $\sigma$ is a
  ground substitution with $\dom(\sigma)=
  (\var(l)\backslash\var(r,\ol{s_n},\ol{t_n}))\cup \bigcup_{i=1}^n
  \var(t_i)\backslash\var(r,\ol{s_\interval{i+1}{n}})$, $\beta:l\to
  r\Leftarrow \ol{s_n\tto t_n}\in\cR$, and $\ol{\hs_n}$ are safe too.
  The pair $\tuple{s,\hs}$ is \emph{safe} in $\cR$ iff $\hs$ is safe.
\end{definition}
Reversible (conditional) rewriting can now be introduced as follows:

\begin{definition}[reversible rewriting]
  Let $\cR$ be a DCTRS. The reversible rewrite relation $\rh_\cR$ is
  defined on safe pairs $\tuple{t,\hs}$, where $t$ is a ground term and
  $\hs$ is a trace in $\cR$. The reversible rewrite relation extends
  standard conditional rewriting as follows:
  \[
  \tuple{s,\hs} \rh_\cR
  \tuple{t,\beta(p,\sigma',\hs_1,\ldots,\hs_n):\hs}
  \]
  iff there exist a position $p\in\pos(s)$, a rewrite rule $\beta: l
  \to r \Leftarrow \ol{s_n\tto t_n} \in \cR$, and a ground
  substitution $\sigma$ such that $s|_p = l\sigma$,
  $\tuple{s_i\sigma,\nil}\rh_\cR^\ast \tuple{t_i\sigma,\hs_i}$ for all
  $i=1,\ldots,n$, $t = s[r\sigma]_p$, and $\sigma' =
  \sigma\!\!\res_{(\var(l)\backslash\var(r,\ol{s_n},\ol{t_n}))\cup
    \bigcup_{i=1}^n
    \var(t_i)\backslash\var(r,\ol{s_\interval{i+1}{n}})}$.
  The reverse relation, $\lh_\cR$, is then defined as
  follows: 
  \[
  \tuple{t,\beta(p,\sigma',\hs_1,\ldots,\hs_n):\hs} \lh_\cR
  \tuple{s,\hs}
  \]
  iff $\tuple{t,\beta(p,\sigma',\ol{\hs_n}):\hs}$ is a safe pair in
  $\cR$, $\beta:l\to r\Leftarrow \ol{s_n\tto t_n}\in\cR$ and there is
  a ground substitution $\theta$ such that
  $\dom(\theta)=\var(r,\ol{s_n})\backslash\dom(\sigma')$, $t|_p =
  r\theta$, $\tuple{t_i\theta\sigma',\hs_i} \lh_\cR^\ast
  \tuple{s_i\theta\sigma',\nil}$ for all $i=1,\ldots,n$, and $s =
  t[l\theta\sigma']_p$.
  Note that $\theta\sigma' = \sigma'\theta = \theta\cup\sigma'$ since
  $\dom(\theta)\cap\dom(\sigma')=\emptyset$ and both substitutions are
  ground.

  As in the unconditional case, we denote the union of both relations
  $\rh_\cR\cup\lh_\cR$ by $\rlh_\cR$.
\end{definition}

\begin{example} 
  Consider again the DCTRS $\cR$ from Example~\ref{ex:dctrs}:
  \[
  \begin{array}{l@{~}r@{~}c@{~}l@{}l@{~}r@{~}c@{~}l}
   \beta_1: & \mathsf{add}(\mathsf{0},y) & \to & y  &
    \beta_4: & \mathsf{even}(\mathsf{0}) & \to & \mathsf{true}\\
    \beta_2: &\mathsf{add}(\mathsf{s}(x),y) & \to &
    \mathsf{s}(\mathsf{add}(x,y)) &
    \beta_5: & \mathsf{even}(\mathsf{s}(\mathsf{s}(x))) & \to & \mathsf{even}(x)\\
    \beta_3: & \mathsf{double}(x) & \to & \mathsf{add}(x,x)
    \Leftarrow \mathsf{even}(x)\tto\mathsf{true}\\
  \end{array}
  \]
  Given the term
  $\mathsf{double}(\mathsf{s}(\mathsf{s}(\mathsf{0})))$, we have, for
  instance, the following forward derivation:
  \[
  \begin{array}{lll}
    \tuple{\mathsf{double}(\mathsf{s}(\mathsf{s}(\mathsf{0}))),{\nil}} \\
    \mbox{}\hspace{10ex} \rh_\cR ~ 
    \tuple{\mathsf{add}(\mathsf{s}(\mathsf{s}(\mathsf{0})),\mathsf{s}(\mathsf{s}(\mathsf{0}))),{[\beta_3(\epsilon,\id,{\hs})]}} \\   
    \mbox{}\hspace{10ex} \rh_\cR ~ \cdots \\
    \mbox{}\hspace{10ex} \rh_\cR ~
    \tuple{\mathsf{s}(\mathsf{s}(\mathsf{s}(\mathsf{s}(\mathsf{0})))),
      {[\beta_1(1.1,\id),\beta_2(1,\id),\beta_2(\epsilon,\id),\beta_3(\epsilon,\id,{\hs})]}} \\   
  \end{array}
  \]
  where ${\hs} = {[\beta_4(\epsilon,\id),\beta_5(\epsilon,\id)]}$ since we
  have the following reduction:
  \[
    \tuple{\mathsf{even}(\mathsf{s}(\mathsf{s}(\mathsf{0}))),{\nil}}
    \rh_\cR 
    \tuple{\mathsf{even}(\mathsf{0}),{[\beta_5(\epsilon,\id)]}} 
     \rh_\cR 
    \tuple{\mathsf{true},{[\beta_4(\epsilon,\id),\beta_5(\epsilon,\id)]}}
  \]  
  The reader can easily construct the associated backward derivation:
  \[
  \tuple{\mathsf{add}(\mathsf{s}(\mathsf{s}(\mathsf{0})),\mathsf{s}(\mathsf{s}(\mathsf{0}))),
    {[\beta_1(1.1,\id),\beta_2(1,\id),\ldots]}}    %%\beta_2(\epsilon,\id),\beta_3(\epsilon,\id,{\hs})]}}
  \lh_\cR^\ast
  \tuple{\mathsf{double}(\mathsf{s}(\mathsf{s}(\mathsf{0}))),{\nil}}
\]
\end{example}
Let us now explain why we need to store $\sigma'$ in a step of the
form $\tuple{s,\hs} \rh_\cR
\tuple{t,\beta(p,\sigma',\ol{\hs_n}):\hs}$.
Given a DCTRS, for each rule $l\to r \Leftarrow \ol{s_n\tto t_n}$, the
following conditions hold:
\begin{itemize}
\item 3-CTRS: $\var(r) \subseteq
  \var(l,\ol{s_n},\ol{t_n})$.
\item Determinism: for all $i=1,\ldots,n$, we have
  $\var(s_i)\subseteq\var(l,\ol{t_{i-1}})$. 
\end{itemize}
Intuitively, the backward relation $\lh_\cR$ can be seen as equivalent
to the forward relation $\rh_\cR$ but using a reverse rule of the form
$r\to l \Leftarrow t_n\tto s_n ,\ldots,t_1\tto s_1$.  Therefore, in
order to ensure that backward reduction is deterministic, we need the
same conditions as above but on the reverse rule:\footnote{We note
  that the notion of a non-erasing rule is extended to the DCTRSs in
  \cite{NSS12lmcs}, which results in a similar condition.}
\begin{itemize}
\item 3-CTRS: $\var(l) \subseteq \var(r,\ol{s_n},\ol{t_n})$.
\item Determinism: for all $i=1,\ldots,n$,
$\var(t_i)\subseteq\var(r,\ol{s_\interval{i+1}{n}})$.
\end{itemize}
Since these conditions cannot be guaranteed in general, we store
\[
\sigma' =
\sigma\!\!\res_{(\var(l)\backslash\var(r,\ol{s_n},\ol{t_n}))\cup
  \bigcup_{i=1}^n
  \var(t_i)\backslash\var(r,\ol{s_\interval{i+1}{n}})}
\]
in the trace term so that $(r\to l \Leftarrow t_n\tto s_n
,\ldots,t_1\tto s_1)\sigma'$ is deterministic and fulfills the
conditions of a 3-CTRS by construction, i.e., $\var(l\sigma')
\subseteq \var(r\sigma',\ol{s_n\sigma'},\ol{t_n\sigma'})$ and for all
$i=1,\ldots,n$,
$\var(t_i\sigma')\subseteq\var(r\sigma',\ol{s_\interval{i+1}{n}\sigma'})$;
see the proof of
Theorem~\ref{th:deterministic-cond} %%in \cite{NPV16b}
for more details.

\begin{example} \label{ex:needforvbles}
  Consider the following DCTRS:
  \[
  \begin{array}{l}
  \beta_1: ~ \mathsf{f}(x,y,m) ~ \to ~ \mathsf{s}(w) \Leftarrow \mathsf{h}(x)\tto
  x,\mathsf{g}(y,\mathsf{4})\tto w\\
  \beta_2: ~\mathsf{h}(\mathsf{0}) ~ \to ~ \mathsf{0} ~~~~~~~~~
  \beta_3: ~\mathsf{h}(\mathsf{1}) ~ \to ~ \mathsf{1} ~~~~~~~~~
  \beta_4: ~ \mathsf{g}(x,y) ~ \to ~ x \\
  \end{array}
  \]
  and the step
  $\tuple{\mathsf{f}(\mathsf{0},\mathsf{2},\mathsf{4}),\nil} \rh_\cR
  \tuple{\mathsf{s}(\mathsf{2}),[\beta_1(\epsilon,\sigma',\hs_1,\hs_2)]}$
  with $\sigma' = \{m\mapsto\mathsf{4},x\mapsto\mathsf{0}\}$, $\hs_1 =
  [\beta_2(\epsilon,\id)]$ and $\hs_2 =
  [\beta_4(\epsilon,\{y\mapsto\mathsf{4}\})]$.  The binding of
  variable $m$ is required to recover the value of the \emph{erased}
  variable $m$, but the binding of variable $x$ is also needed to
  perform the subderivation $\tuple{x,\pi_1} \lh_\cR
  \tuple{\mathsf{h}(x),\nil}$ when applying a backward step from
  $\tuple{\mathsf{s}(\mathsf{2}),[\beta_1(\epsilon,\sigma',\hs_1,\hs_2)]}$. If
  the binding for $x$ were unknown, this step would not be
  deterministic. As mentioned above, an instantiated reverse rule
  $(\mathsf{s}(w) \to \mathsf{f}(x,y,m) \Leftarrow w \tto
  \mathsf{g}(y,\mathsf{4}), x \tto\mathsf{h}(x))\sigma' ~=~
  \mathsf{s}(w) \to \mathsf{f}(\mathsf{0},y,\mathsf{4}) \Leftarrow w
  \tto \mathsf{g}(y,\mathsf{4}),\mathsf{0} \tto\mathsf{h}(\mathsf{0})$
  would be a legal DCTRS rule thanks to $\sigma'$.
\end{example}
We note that similar conditions could be defined for arbitrary
3-CTRSs. However, the conditions would be much more involved; e.g.,
one had to compute first the \emph{variable dependencies} between the
equations in the conditions. Therefore, we prefer to keep the simpler
conditions for DCTRSs (where these dependencies are fixed), which is
still quite a general class of CTRSs.

Reversible rewriting is also a conservative extension of rewriting for
DCTRSs (we omit the proof since it is straightforward):

\begin{theorem} \label{th:conservative} Let $\cR$ be a DCTRS. Given
  ground terms $s,t$, if $s\to_\cR^\ast t$, then for any trace $\hs$
  there exists a trace $\hs'$ such that $\tuple{s,\hs} \rh_\cR^\ast
  \tuple{t,\hs'}$.
\end{theorem}
For the following result, we need some preliminary notions (see, e.g.,
\cite{Terese03}).  For every oriented CTRS $\cR$, we inductively
define the TRSs $\cR_k$, $k\geq 0$, as follows:
\[
\begin{array}{l@{~}l@{~}l}
  \cR_0 & = & \emptyset \\
  \cR_{k+1} & = & \{ l\sigma \to r\sigma\mid l\to r\Leftarrow
  \ol{s_n\tto t_n}\in\cR, ~
  s_i\sigma
  \to^\ast_{\cR_k} t_i\sigma~\mbox{for all}~i=1,\ldots,n\}
\end{array}
\]
Observe that $\cR_k\subseteq\cR_{k+1}$ for all $k\geq 0$. We have
$\to_\cR = \bigcup_{i\geq 0} \to_{\cR_i}$.  We also have $s\to_\cR t$
iff $s\to_{\cR_k} t$ for some $k\geq 0$. The minimum such $k$ is
called the \emph{depth} of $s\to_\cR t$, and the maximum depth $k$ of
$s = s_0 \to_{\cR_{k_1}} %s_1 \to_{\cR_{k_2}}
\cdots \to_{\cR_{k_m}} s_m = t$ (i.e., $k$ is the maximum of depths
$k_1,\ldots,k_m$) is called the \emph{depth} of the derivation.
If a derivation has depth $k$ and length $m$, we write $s \to_{\cR_k}^m t$.
Analogous notions can naturally be defined for $\rh_\cR$, $\lh_\cR$,
and $\rlh_\cR$.

The next result shows that safe pairs are also preserved through
reversible rewriting with DCTRSs:

\begin{lemma} \label{lemma:safe-cond} Let $\cR$ be a DCTRS and
  $\tuple{s,\hs}$ a safe pair. If $\tuple{s,\hs} \rlh_\cR^\ast
  \tuple{t,\hs'}$, then $\tuple{t,\hs'}$ is also safe.
\end{lemma}

\begin{proof} 
  We prove the claim by induction on the lexicographic product $(k,m)$
  of the depth $k$ and the length $m$ of the derivation $\tuple{s,\hs}
  \rlh_{\cR_k}^m \tuple{t,\hs'}$.  Since the base case is trivial, we
  consider the inductive case $(k,m)>(0,0)$. Consider a derivation
  $\tuple{s,\hs} \rlh_{\cR_k}^{m-1} \tuple{s_0,\hs_0} \rlh_{\cR_k}
  \tuple{t,\hs'}$.  By the induction hypothesis, we have that
  $\tuple{s_0,\hs_0}$ is safe. Now, we distinguish two cases depending
  on the last step. If the last step is $\tuple{s_0,\hs_0} \rh_{\cR_k}
  \tuple{t,\hs'}$, then there exist a position $p\in\pos(s_0)$, a
  rewrite rule $\beta: l \to r \Leftarrow \ol{s_n\tto t_n} \in \cR$,
  and a ground substitution $\sigma$ such that $s_0|_p = l\sigma$,
  $\tuple{s_i\sigma,\nil}\rh_{\cR_{k_i}}^\ast \tuple{t_i\sigma,\hs_i}$
  for all $i=1,\ldots,n$, $t = s_0[r\sigma]_p$, $\sigma' =
  \sigma\!\!\res_{(\var(l)\backslash\var(r,\ol{s_n},\ol{t_n}))\cup
    \bigcup_{i=1}^n
    \var(t_i)\backslash\var(r,\ol{s_\interval{i+1}{n}})}$, and $\hs' =
  \beta(p,\sigma',\hs_1,\ldots,\hs_n)$. Then, since $k_i<k$,
  $i=1,\ldots,n$, $\sigma'$ is ground and $\dom(\sigma')=
  (\var(l)\backslash\var(r,\ol{s_n},\ol{t_n}))\cup \bigcup_{i=1}^n
  \var(t_i)\backslash\var(r,\ol{s_\interval{i+1}{n}})$ by
  construction, the claim follows by induction.
  Finally, if the last step has the form $\tuple{s_0,\hs_0}
  \lh_{\cR_k} \tuple{t,\hs'}$, then the claim follows trivially since
  a step with $\lh_\cR$ only removes trace terms from $\hs_0$. \qed
\end{proof}
As in the unconditional case, the following proposition follows
straightforwardly from the previous lemma since any pair with an empty
trace is safe.

\begin{proposition} \label{th:safe2} Let $\cR$ be a DCTRS. If
  $\tuple{s,\nil} \rlh_\cR^\ast \tuple{t,\pi}$, then $\tuple{t,\pi}$
  is safe in $\cR$.
\end{proposition}
Now, we can already state the reversibility of $\rh_\cR$ for DCTRSs:

\begin{theorem} \label{th:conditional} Let $\cR$ be a DCTRS. Given the
  safe pairs $\tuple{s,\hs}$ and $\tuple{t,\hs'}$, for all $k,m\geq
  0$, $\tuple{s,\hs} \rh_{\cR_k}^m \tuple{t,\hs'}$ iff $\tuple{t,\hs'}
  \lh_{\cR_k}^m \tuple{s,\hs}$.
%  for some $k,m\geq 0$.
\end{theorem}

\begin{proof}
  $(\Rightarrow)$ We prove the claim by induction on the lexicographic
  product $(k,m)$ of the depth $k$ and the length $m$ of the
  derivation $\tuple{s,\hs} \rh_{\cR_k}^m \tuple{t,\hs'}$. Since the base
  case is trivial, we consider the inductive case
  $(k,m)>(0,0)$. Consider a derivation
  $\tuple{s,\hs} \rh_{\cR_k}^{m-1} \tuple{s_0,\hs_0} \rh_{\cR_k}
  \tuple{t,\hs'}$
  whose associated product is $(k,m)$. By Proposition~\ref{th:safe2},
  both $\tuple{s_0,\hs_0}$ and $\tuple{t,\hs'}$ are safe. By the
  induction hypothesis, since $(k,m-1)<(k,m)$, we have
  $\tuple{s_0,\hs_0} \lh_{\cR_k}^{m-1} \tuple{s,\hs}$. Consider now
  the step $\tuple{s_0,\hs_0} \rh_{\cR_k} \tuple{t,\hs'}$. Thus, there
  exist a position $p\in\pos(s_0)$, a rule
  $\beta:l\to r \Leftarrow \ol{s_n\tto t_n}\in\cR$, and a ground substitution
  $\sigma$ such that $s_0|_p=l\sigma$,
  $\tuple{s_i\sigma,\nil}\rh_{\cR_{k_i}}^\ast
  \tuple{t_i\sigma,\hs_i}$
  for all $i=1,\ldots,n$, $t=s_0[r\sigma]_p$,
  $\sigma'=\sigma\!\!\res_{(\var(l)\backslash\var(r,\ol{s_n},\ol{t_n}))\cup
    \bigcup_{i=1}^n
    \var(t_i)\backslash\var(r,\ol{s_\interval{i+1}{n}})}$,
  and $\hs'=\beta(p,\sigma',\hs_1,\ldots,\hs_n):\hs_0$.  By definition
  of $\rh_{\cR_k}$, we have that $k_i<k$ and, thus,
  $(k_i,m_1)<(k,m_2)$ for all $i=1,\ldots,n$ and for all
  $m_1,m_2$. Hence, by the induction hypothesis, we have
  $\tuple{t_i\sigma,\hs_i}\lh_{\cR_{k_i}}^\ast
  \tuple{s_i\sigma,\nil}$ for all $i=1,\ldots,n$.
  Let $\theta =
  \sigma\!\res_{\var(r,\ol{s_n})\backslash\dom(\sigma')}$, so that
  $\sigma=\theta\sigma'$ and
  $\dom(\theta)\cap\dom(\sigma')=\emptyset$. Therefore, we have
  $\tuple{t,\hs'} \lh_{\cR_k} \tuple{s'_0,\hs_0}$ with $t|_p=r\theta$,
  $\beta:l\to r\Leftarrow \ol{s_n\tto t_n}\in\cR$ and
  $s'_0=t[l\theta\sigma']_p = t[l\sigma]_p = s_0$, and the claim
  follows.

  $(\Leftarrow)$ This direction proceeds in a similar way. We prove
  the claim by induction on the lexicographic product $(k,m)$ of the
  depth $k$ and the length $m$ of the considered derivation. Since the
  base case is trivial, let us consider the inductive case
  $(k,m)>(0,0)$.
  Consider a derivation $\tuple{t,\hs'} \lh_{\cR_k}^{m-1}
  \tuple{s_0,\hs_0} \lh_{\cR_k} \tuple{s,\hs}$ whose associated
  product is $(k,m)$. By Proposition~\ref{th:safe2}, both
  $\tuple{s_0,\hs_0}$ and $\tuple{s,\hs}$ are safe. By the induction
  hypothesis, since $(k,m-1)<(k,m)$, we have $\tuple{s_0,\hs_0}
  \rh_{\cR_k}^{m-1} \tuple{t,\hs'}$. Consider now the step
  $\tuple{s_0,\hs_0} \lh_{\cR_k} \tuple{s,\hs}$. Then, we have $\hs_0
  = \beta(p,\sigma',\hs_1,\ldots,\hs_n):\hs$, $\beta:l \to r\Leftarrow
  \ol{s_n\tto t_n}\in\cR$, and there exists a ground substitution
  $\theta$ with $\dom(\theta)=\var(r,\ol{s_n})\backslash\dom(\sigma')$
  such that $s_0|_p= r\theta$, $\tuple{t_i\theta\sigma',\hs_i}
  \lh_{\cR_{k_i}}^\ast \tuple{s_i\theta\sigma',\nil} $ for all
  $i=1,\ldots,n$, and $s=s_0[l\theta\sigma']_p$. Moreover, since
  $\tuple{s_0,\hs_0}$ is safe, we have that
  $\dom(\sigma')=(\var(l)\backslash\var(r,\ol{s_n},\ol{t_n}))\cup
  \bigcup_{i=1}^n \var(t_i)\backslash\var(r,\ol{s_\interval{i+1}{n}})$.
  By definition of $\lh_{\cR_k}$, we have that $k_i<k$ and, thus,
  $(k_i,m_1)<(k,m_2)$ for all $i=1,\ldots,n$ and for all
  $m_1,m_2$. By the induction hypothesis, we have
  $\tuple{s_i\theta\sigma',\nil}\rh_{\cR_{k_i}}^\ast
  \tuple{t_i\theta\sigma',\hs_i}$ for all $i=1,\ldots,n$.  Let
  $\sigma = \theta\sigma'$, with
  $\dom(\theta)\cap\dom(\sigma')=\emptyset$. Then, since
  $s|_p = l\sigma$, we can perform the step
  $\tuple{s,\hs} \rh_{\cR_k}
  \tuple{s'_0,\beta(p,\sigma',\hs_1,\ldots,\hs_n):\hs}$ with
  $s'_0 = s[r\sigma]_p = s[r\theta\sigma']_p$; moreover,
  $s[r\theta\sigma']_p= s[r\theta]_p = s_0[r\theta]_p = s_0$ since
  $\dom(\sigma')\cap\var(r)=\emptyset$, which concludes the
  proof. \qed
\end{proof}
In the following, we say that $\tuple{t,\hs'} \lh_\cR \tuple{s,\hs}$
is a \emph{deterministic} step if there is no other, different pair
$\tuple{s'',\hs''}$ with $\tuple{t,\hs'} \lh_\cR \tuple{s'',\hs''}$
and, moreover, the subderivations for the equations in the condition
of the applied rule (if any) are deterministic, too. We say that a
derivation $\tuple{t,\hs'} \lh_\cR^\ast \tuple{s,\hs}$ is
deterministic if each reduction step in the derivation is
deterministic.

Now, we can already prove that backward reversible rewriting is also
deterministic, as in the unconditional case:

\begin{theorem} \label{th:deterministic-cond} Let $\cR$ be a
  DCTRS. Let $\tuple{t,\hs'}$ be a safe pair with $\tuple{t,\hs'}
  \lh^\ast_\cR \tuple{s,\hs}$ for some term $s$ and trace $\hs$. Then
  $\tuple{t,\hs'} \lh^\ast_\cR \tuple{s,\hs}$ is deterministic.
\end{theorem}

\begin{proof}
  We prove the claim by induction on the lexicographic product $(k,m)$
  of the depth $k$ and the length $m$ of the steps.  The case $m=0$ is
  trivial, and thus we let $m > 0$.  Assume $\tuple{t,\hs'}
  \lh^{m-1}_{\cR_k} \tuple{u,\hs''} \lh_{\cR_k} \tuple{s,\hs}$.  For
  the base case $k=1$, the applied rule is unconditional and the proof
  is analogous to that of Theorem~\ref{th:deterministic}.

  Let us now consider $k>1$. By definition, if $\tuple{u,\hs''}
  \lh_{\cR_k} \tuple{s,\hs}$, we have
  $\hs''=\beta(p,\sigma',\hs_1,\ldots,\hs_n):\hs$, $\beta:l\to
  r\Leftarrow \ol{s_n\tto t_n}\in\cR$ and there exists a ground
  substitution $\theta$ with $\dom(\theta)=\var(r)$ such that $u|_p =
  r\theta$, $\tuple{t_i\theta\sigma',\hs_i} \lh_{\cR_j}^\ast
  \tuple{s_i\theta\sigma',\nil}$, $j<k$, for all $i=1,\ldots,n$, and
  $s = t[l\theta\sigma']_p$.  By the induction hypothesis, the
  subderivations $\tuple{t_i\theta\sigma',\hs_i} \lh_{\cR_j}^\ast
  \tuple{s_i\theta\sigma',\nil}$ are deterministic, i.e.,
  $\tuple{s_i\theta\sigma',\nil}$ is a unique resulting term obtained
  by reducing $\tuple{t_i\theta\sigma',\hs_i}$. Therefore, the only
  remaining source of nondeterminism can come from choosing a rule
  labeled with $\beta$ and from the computed substitution $\theta$. On
  the one hand, the labels are unique in $\cR$.  As for $\theta$, we
  prove that this is indeed the only possible substitution for the
  reduction step. Consider the instance of rule $l\to r\Leftarrow
  \ol{s_n\tto t_n}$ with $\sigma'$: $l\sigma'\to r\sigma'\Leftarrow
  \ol{s_n\sigma'\tto t_n\sigma'}$. Since $\tuple{u,\hs''}$ is safe, we
  have that $\sigma'$ is a ground substitution and $\dom(\sigma')=
  (\var(l)\backslash\var(r,\ol{s_n},\ol{t_n}))\cup \bigcup_{i=1}^n
  \var(t_i)\backslash\var(r,\ol{s_\interval{i+1}{n}})$. Then, the
  following properties hold:
  \begin{itemize}
  \item $\var(l\sigma') \subseteq
    \var(r\sigma',\ol{s_n\sigma'},\ol{t_n\sigma'})$, since $\sigma'$
    is ground and it covers all the variables in
    $\var(l)\backslash\var(r,\ol{s_n},\ol{t_n})$.

  \item
    $\var(t_i\sigma')\subseteq\var(r\sigma',\ol{s_\interval{i+1}{n}\sigma'})$
    for all $i=1,\ldots,n$, since $\sigma'$ is ground and it covers
    all variables in $\bigcup_{i=1}^n
    \var(t_i)\backslash\var(r,\ol{s_\interval{i+1}{n}})$.
  \end{itemize}
  The above properties guarantee that a rule of the form $r\sigma'\to
  l\sigma'\Leftarrow t_n\sigma'\tto s_n\sigma',\ldots,t_1\sigma'\tto
  s_1\sigma'$ can be seen as a rule of a DCTRS and, thus, there exists
  a deterministic procedure to compute $\theta$, which completes the
  proof. \qed
\end{proof}
Therefore, $\lh_\cR$ is deterministic and confluent. Termination is
trivially guaranteed for pairs with a finite trace since the trace's
length strictly decreases with every backward step.

%%%%%%%%%%%%%%%%%%%%%%%%%%%%%%%%%%%%%%%%%%%%%%%%%%%%%%%%%%%%%%%%%%
\section{Removing Positions from Traces} \label{sec:irr}

Once we have a feasible definition of reversible rewriting, there are
two refinements that can be considered: i) reducing the size of the
traces and ii) defining a \emph{reversibilization} transformation so
that standard rewriting becomes reversible in the transformed
system. In this section, we consider the first problem, leaving the
second one for the next section. 

In principle, one could remove information from the traces by
requiring certain conditions on the considered systems. For instance,
requiring injective functions may help to remove rule labels from
trace terms. Also, requiring \emph{non-erasing} rules may help to
remove the second component of trace terms (i.e., the
substitutions). In this section, however, we deal with a more
challenging topic: removing positions from traces. This is useful not
only to reduce the size of the traces but it is also essential to
define a reversibilization technique for DCTRSs in the next
section.\footnote{We note that defining a transformation with traces
  that include positions would be a rather difficult task because
  positions are \emph{dynamic} (i.e., they depend on the term being
  reduced) and thus would require a complex (and inefficient) system
  instrumentation.}
In particular, we aim at transforming a given DCTRS into one that
fulfills some conditions that make storing positions unnecessary.

In the following, given a CTRS $\cR$, we say that a term $t$ is
\emph{basic} \cite{HM08} if it has the form $f(\ol{t_n})$ with
$f\in\cD_\cR$ a defined function symbol and
$\ol{t_n}\in\cT(\cC_\cR,\cV)$ constructor terms.
Furthermore, in the remainder of this paper, we assume that the
right-hand sides of the equations in the conditions of the rules of a
DCTRS are constructor terms. This is not a significant restriction
since these terms cannot be reduced anyway (since we consider oriented
equations in this paper), and still covers most practical examples.

Now, we introduce the following subclass of DCTRSs:

\begin{definition}[\pc\ \cite{NSS12}]
  We say that a DCTRS $\cR$ is a \pc\ (``pc'' stands for \emph{pure
    constructor}) if, for each rule
  $l\to r \Leftarrow \ol{s_n\tto t_n}\in \cR$, we have that $l$ and
  $\ol{s_n}$ are basic terms and $r$ and $\ol{t_n}$ are constructor
  terms.
\end{definition}
Pure constructor systems are called \emph{normalized} systems in
\cite{AV06}. Also, they are mostly equivalent to the class III$_n$ of
conditional systems in \cite{BK86}, where $t_1,\ldots,t_n$ are
required to be ground unconditional normal forms
instead.\footnote{Given a CTRS $\cR$, we define $\cR_u = \{ l\to r
  \mid l\to r \Leftarrow \ol{s_n\tto t_n}\in\cR\}$. A term is an
  \emph{unconditional} normal form in $\cR$, if it is a normal form in
  $\cR_u$.}

In principle, any DCTRS with basic terms in the left-hand sides (i.e.,
a \emph{constructor} DCTRS) and constructor terms in the right-hand
sides of the equations of the rules can be transformed into a \pc\ by
applying a few simple transformations: flattening and simplification
of constructor conditions.
Let us now consider each of these transformations separately.  Roughly
speaking, flattening involves transforming a term (occurring, e.g., in
the right-hand side of a DCTRS or in the condition) with nested
defined functions like $\mathsf{f}(\mathsf{g}(x))$ into a term
$\mathsf{f}(y)$ and an (oriented) equation $\mathsf{g}(x)\tto y$,
where $y$ is a fresh variable. Formally,

\begin{definition}[flattening]
  Let $\cR$ be a CTRS, $R = (l\to r\Leftarrow\ol{s_n\tto t_n})\in\cR$
  be a rule and $R'$ be a new rule either of the form $l\to r
  \Leftarrow s_1\tto t_1,\ldots,s_i|_p \tto w,s_i[w]_p\tto
  t_i,\ldots,s_n\tto t_n$, for some $p\in\pos(s_i)$, $1\sleq i\sleq
  n$, or $l\to r[w]_q \Leftarrow \ol{s_n\tto t_n},r|_q\tto w$, for
  some $q\in\pos(r)$, where $w$ is a fresh variable.\footnote{The
    positions $p,q$ can be required to be different from $\toppos$,
    but this is not strictly necessary.}
  Then, a CTRS $\cR'$ is obtained from $\cR$ by a \emph{flattening}
  step if $\cR' = (\cR\backslash\{R\})\cup\{R'\}$.
\end{definition}
Note that, if an unconditional rule is non-erasing (i.e.,
$\var(l)\subseteq\var(r)$ for a rule $l\to r$), any conditional rule
obtained by flattening is trivially non-erasing too, according to the
notion of non-erasingness for DCTRSs in
\cite{NSS12lmcs}.\footnote{Roughly, a DCTRS is considered non-erasing
  in \cite{NSS12lmcs} if its transformation into an unconditional TRS
  by an unraveling transformation gives rise to a non-erasing TRS.}

Flattening is trivially \emph{complete} since any flattening step can
be undone by binding the fresh variable again to the selected subterm
and, then, proceeding as in the original system. Soundness is more
subtle though. In this work, we prove the correctness of flattening
for arbitrary DCTRSs with respect to \emph{innermost} rewriting.  As usual,
the innermost rewrite relation, in symbols, $\inn_\cR$, is defined as
the smallest binary relation satisfying the following: given ground
terms $s,t\in\cT(\cF)$, we have $s \inn_\cR t$ iff there exist a
position $p$ in $s$ such that no proper subterms of $s|_p$ are
reducible, a rewrite rule $l \to r\Leftarrow \ol{s_n\tto t_n} \in
\cR$, and a normalized ground substitution $\sigma$ such that $s|_p =
l\sigma$, $s_i\sigma \inns_\cR \: t_i\sigma$,
for all $i=1,\ldots,n$, and $t = s[r\sigma]_p$.

In order to prove the correctness of flattening, we state the
following auxiliary lemma:

\begin{lemma} \label{lemma:flatening} Let $\cR$ be a DCTRS. Given
  terms $s$ and $t$, with $t$ a normal form, and a position $p\in\pos(s)$, we have
  $s \inns_\cR \: t$ iff $s|_p \inns_\cR \: w\sigma$ and
  $s[w\sigma]_p\inns_\cR \:t$, for some fresh variable $w$ and
  normalized substitution $\sigma$.
\end{lemma}

\begin{proof}
  $(\Rightarrow)$ Let us consider an arbitrary position
  $p\in\pos(s)$. If $s|_p$ is normalized, the proof is
  straightforward. Otherwise, since we use innermost reduction
  (\emph{leftmost} innermost, for simplicity), we can represent the
  derivation $s \inns_\cR \: t$ as follows:
  \[
    s[s|_p]_p \inns_\cR \: s'[s|_p]_p \inns_\cR \: s'[s'']_p \inns_\cR
    \: t
  \]
  where $s''$ is a normal form and the subderivation $s[s|_p]_p \inns_\cR \: s'[s|_p]_p$ reduces
  the leftmost innermost subterms that are to the left of $s|_p$ (if
  any). Then, by choosing $\sigma = \{w\mapsto s''\}$ we have
  $s|_p \inns_\cR \: w\sigma$ (by mimicking the steps of
  $s'[s|_p]_p \inns_\cR \: s'[s'']_p$),
  $s[w\sigma]_p {\inns_\cR} \: s'[w\sigma]_p$ (by mimicking the steps of
  $s[s|_p]_p {\inns_\cR} \: s'[s|_p]_p$), and
  $s'[w\sigma]_p {\inns_\cR} \: t$ (by mimicking the steps of
  $s'[s'']_p \inns_\cR \: t$), which concludes the proof.

  $(\Leftarrow)$ This direction is perfectly analogous to the previous
  case. We consider an arbitrary position $p\in\pos(s)$ such that
  $s|_p$ is not normalized (otherwise, the proof is trivial). Now,
  since derivations are innermost, we can consider that
  $s[w\sigma]_p\inns_\cR \:t$ is as follows:
  $s[w\sigma]_p\inns_\cR \: s'[w\sigma]_p \inns_\cR \:t$, where
  $s[w\sigma]_p\inns_\cR \: s'[w\sigma]_p$ reduces the innermost
  subterms to the left of $s|_p$. Therefore, we have
  $s[s|_p]_p \inns_\cR \: s'[s|_p]_p$ (by mimicking the steps of
  $s[w\sigma]_p\inns_\cR \: s'[w\sigma]_p$),
  $s'[s|_p]_p \inns_\cR \: s'[s'']_p$ (by mimicking the steps of
  $s|_p \inns_\cR \: w\sigma$, with $\sigma=\{w\mapsto s''\}$), and
  $s'[s'']_p \inns_\cR \: t$ (by mimicking the steps of
  $s'[w\sigma]_p \inns_\cR \:t$). %Note that
\qed
\end{proof}
The following theorem is an easy consequence of the previous
lemma:

\begin{theorem} \label{th:flattening} Let $\cR$ be a DCTRS. If $\cR'$
  is obtained from $\cR$ by a flattening step, then $\cR'$ is a DCTRS
  and, for all ground terms $s,t$, with $t$ a normal form,
  we have $s \inns_\cR \:t$ iff $s \inns_{\cR'} \:t$.
\end{theorem}

\begin{proof}
  $(\Rightarrow)$ We prove the claim by induction on the lexicographic
  product $(k,m)$ of the depth $k$ and the length $m$ of the
  derivation $s \inns_{\cR_k} \:t$.  Since the base case is trivial,
  we consider the inductive case $(k,m)>(0,0)$.  Assume that $s
  \inns_{\cR_k} \:t$ has the form $s[l\sigma]_u \inn_{\cR_k} \:
  s[r\sigma]_u \inns_{\cR_k}\: t$ with $l\to r \Leftarrow \ol{s_n\tto
    t_n} \in \cR$ and $s_i\sigma \inns_{\cR_{k_i}} \;t_i\sigma$,
  $k_i<k$, $i=1,\ldots,n$. If $l\to r \Leftarrow \ol{s_n\tto t_n} \in
  \cR'$, the claim follows directly by induction. Otherwise, we have
  that either $l\to r \Leftarrow s_1\tto t_1,\ldots,s_i|_p \tto
  w,s_i[w]_p\tto t_i,\ldots,s_n\tto t_n \in \cR'$, for some
  $p\in\pos(s_i)$, $1\sleq i\sleq n$, or $l\to r[w]_q \Leftarrow
  \ol{s_n\tto t_n},r|_q\tto w\in\cR'$, for some $q\in\pos(r)$, where
  $w$ is a fresh variable.
  Consider first the case
  $l\to r \Leftarrow s_1\tto t_1,\ldots,s_i|_p \tto w,s_i[w]_p\tto
  t_i,\ldots,s_n\tto t_n \in \cR'$, for some $p\in\pos(s_i)$,
  $1\sleq i\sleq n$. Since $s_i\sigma \inns_{\cR_{k_i}} \;t_i\sigma$,
  $k_i<k$, $i=1,\ldots,n$, by the induction hypothesis, we have
  $s_i\sigma \inns_{\cR'} \;t_i\sigma$, $i=1,\ldots,n$. By
  Lemma~\ref{lemma:flatening}, there exists
  $\sigma' = \{w\mapsto s'\}$ for some normal form $s'$ such that
  $s_i|_p\sigma = s_i|_p\sigma\sigma' \inns_{\cR_{k_i}} \:
  w\sigma\sigma' = w\sigma'$ and
  $s_i[w]_p\sigma\sigma' = s_i\sigma[w\sigma']_p \inns_{\cR_{k_i}}
  t_i$. Moreover, since $w$ is an extra variable, we also have
  $s_j\sigma\sigma' = s_j\sigma \inns_{\cR'} \:t_j\sigma =
  t_j\sigma\sigma'$ for $j=1,\ldots,i-1,i+1,\ldots,n$.  Therefore,
  since $l\sigma\sigma' = l\sigma$ and $r\sigma\sigma' = r\sigma$, we
  have $s[l\sigma]_u \inn_{\cR} s[r\sigma]_u$, and the claim follows
  by induction.
  Consider the second case.  By the induction hypothesis, we have
  $s[r\sigma]_u \inns_{\cR'} \: t$ and
  $s_i\sigma\inns_{\cR'}\:t_i\sigma$ for all $i=1,\ldots,n$. By
  Lemma~\ref{lemma:flatening}, there exists a substitution $\sigma' =
  \{ w \mapsto s'\}$ such that $s'$ is the normal form of $r|_q\sigma$
  and we have $r|_q\sigma \inns_{\cR'} \:w\sigma'$ and $s[r\sigma[w\sigma']_q]_u
  \inns_{\cR'} \: t$. Moreover, since $w$ is a fresh variable, we have
  $s_i\sigma\sigma'\inns_{\cR'}\:t_i\sigma\sigma'$ for all
  $i=1,\ldots,n$. Therefore, we have $s[l\sigma\sigma']_u = s[l\sigma]_u
  \inn_{\cR'} s[r\sigma[w\sigma']_q]_u$,
  which concludes the proof.

  $(\Leftarrow)$ This direction is perfectly analogous to the previous
  one, and follows easily by Lemma~\ref{lemma:flatening} too. \qed
\end{proof}
Let us now consider the second kind of transformations: the
simplification of constructor conditions. Basically, we can drop an
equation $s \tto t$ when the terms $s$ and $t$ are constructor, called
a \emph{constructor condition}, by either applying the \emph{most
  general unifier} (mgu) of $s$ and $t$ (if it exists) to the
remaining part of the rule, or by deleting entirely the rule if they
do not unify because (under innermost rewriting) the equation will
never be satisfied by any normalized substitution. Similar
transformations can be found in \cite{NV11rta}.

In order to justify these transformations, we state and prove the
following results. In the following, we let $mgu(s,t)$ denote the most
general unifier of terms $s$ and $t$ if it exists, and $\mathit{fail}$
otherwise.

\begin{theorem}[removal of unifiable constructor conditions]
  \label{thm:remove_constructor-conditions} Let $\cR$ be a DCTRS and
  let $R = (l\to r\Leftarrow\ol{s_n\tto t_n})\in\cR$ be a rule with
  $mgu(s_i,t_i) = \theta$, for some $i\in\{1,\ldots,n\}$, where $s_i$
  and $t_i$ are constructor terms.
  Let $R'$ be a new rule of the form $l\theta \to r\theta \Leftarrow
  s_1\theta \tto t_1\theta,\ldots, s_{i-1}\theta \tto t_{i-1}\theta,
  s_{i+1}\theta \tto t_{i+1}\theta,\ldots,s_n\theta \tto t_n\theta$.%
  \footnote{In~\cite{NV11rta}, the condition $\dom(\theta) \cap
    \Var(l,r,s_1,t_1,\ldots,s_n,t_n)=\emptyset$ is required, but this
    condition is not really necessary. }
  Then $\cR' = (\cR\backslash\{R\})\cup\{R'\}$ is a DCTRS and, for all
  ground terms $s$ and $t$, %%with $t$ a constructor term,
  we have $s \inns_\cR \:t$ iff $s \inns_{\cR'} \:t$.
\end{theorem}

\begin{proof} $(\Rightarrow)$ First, we prove the following claim by
  induction on the lexicographic product $(k,m)$ of the depth $k$ and
  the length $m$ of the steps: if $s \innm_{\cR_k} \:t$, then $s
  \inns_{\cR'} \:t$.  It suffices to consider the case where $R$ is
  applied, i.e., $s = s[l\sigma]_p \inn_{\{R\}} \: s[r\sigma]_p$ with
  $s_j\sigma \inns_{\cR_{k_j}} \: t_j\sigma$ for all $j \in
  \{1,\ldots,n\}$.  By definition, $\sigma$ is normalized. Hence,
  since $s_i$ and $t_i$ are constructor terms, we have that
  $s_i\sigma$ and $t_i\sigma$ are trivially normal forms since the
  normalized subterms introduced by $\sigma$ cannot become reducible
  in a constructor context. Therefore, we have $s_i\sigma =
  t_i\sigma$.  Thus, $\sigma$ is a unifier of $s_i$ and $t_i$ and,
  hence, $\theta$ is more general than $\sigma$.  Let $\delta$ be a
  substitution such that $\sigma = \theta\delta$.  Since $\sigma$ is
  normalized, so is $\delta$.  Since $k_j<k$ for all $j=1,\ldots,n$,
  by the induction hypothesis, we have that $s_j\sigma =
  s_j\theta\delta \inns_{\cR'} \: t_j\theta\delta = t_j\sigma$ for $j
  \in \{1,\ldots,i-1,i+1,\ldots,n\}$.  Therefore, we have that
  $s[l\sigma]_p = s[l\theta\delta]_p \inn_{\{R'\}} \:
  s[r\theta\delta]_p = s[r\sigma]_p$.
 
  $(\Leftarrow)$ Now, we prove the following claim by induction on the
  lexicographic product $(k,m)$ of the depth $k$ and the length $m$ of
  the steps: if $s \innm_{\cR'_{k}} \:t$, then $s \inns_\cR \:t$.  It
  suffices to consider the case where $R'$ is applied, i.e., $s =
  s[l\theta\delta]_p \inn_{\{R\}} \: s[r\theta\delta]_p$ with
  $s_j\theta\delta \inns_{\cR'_{k_j}} \: t_j\theta\delta$ for all $j
  \in \{1,\ldots,i-1,i+1,\ldots,n\}$.  By the assumption and the
  definition, $\theta$ and $\delta$ are normalized, and thus,
  $s_i\theta\delta$ and $t_i\theta\delta$ are normal forms (as in the
  previous case, because the normalized subterms introduced by
  $\theta\delta$ cannot become reducible in a constructor context),
  i.e., $s_i\theta\delta = t_i\theta\delta$.  Since $k_j<k$ for all $j
  \in \{1,\ldots,i-1,i+1,\ldots,n\}$, by the induction hypothesis, we
  have that $s_j\theta\delta \inns_{\cR} \: t_j\theta\delta$ for $j
  \in \{1,\ldots,i-1,i+1,\ldots,n\}$.  Therefore, we have that
  $s[l\sigma]_p = s[l\theta\delta]_p \inn_{\{R\}} \:
  s[r\theta\delta]_p = s[r\sigma]$ with $\sigma=\theta\delta$.  \qed
\end{proof}
Now we consider the case when the terms in the constructor condition
do not unify:

\begin{theorem}[removal of infeasible rules]
  \label{thm:remove_infeasible-rules}
  Let $\cR$ be a DCTRS and let $R = (l\to r\Leftarrow\ol{s_n\tto
    t_n})\in\cR$ be a rule with $mgu(s_i,t_i) = \mathit{fail}$,
  for some $i\in\{1,\ldots,n\}$.
  Then $\cR' = \cR\backslash\{R\}$ is a DCTRS and, for all ground
  terms $s$ and $t$, %%with $t$ a constructor term, 
  we have $s \inns_\cR \:t$ iff $s \inns_{\cR'} \:t$.
\end{theorem}

\begin{proof}
  Since $\cR \supseteq \cR'$, the \emph{if} part is trivial, and thus,
  we consider the \emph{only-if} part.  To apply $R$ to a term, there
  must exist a normalized substitution $\sigma$ such that $s_i\sigma
  \inns_\cR \: t_i\sigma$.  Since $s_i,t_i$ are constructor terms and
  $\sigma$ is normalized, $s_i\sigma$ is a normal form (because the
  normalized subterms introduced by $\sigma$ cannot become reducible
  in a constructor context). If $s_i\sigma \inns_\cR \: t_i\sigma$ is
  satisfied (i.e., $s_i\sigma = t_i\sigma$), then $s_i$ and $t_i$ are
  unifiable, and thus, this contradicts the assumption.  Therefore,
  $R$ is never applied to any term, and hence, $s \inns_\cR \:t$ iff
  $s \inns_{\cR'} \:t$.  \qed
\end{proof}
Using flattening and the simplification of constructor conditions, any
constructor DCTRS with constructor terms in the right-hand sides of
the equations of the rules can be transformed into a \pc. One can use,
for instance, the following simple algorithm. Let $\cR$ be such a
constructor DCTRS. We apply the following transformations as much as
possible:
\begin{description}
\item[\sf (flattening-rhs)] Assume that $\cR$ contains a rule of the
  form $R = (l\to r\Leftarrow\ol{s_n\tto t_n})$ where $r$ is not a
  constructor term. Let $r|_q$, $q\in\pos(r)$, be a basic subterm of
  $r$. Then, we replace rule $R$ by a new rule of the form
  $l\to r[w]_q \Leftarrow \ol{s_n\tto t_n},r|_q\tto w$, where $w$ is a
  fresh variable.

\item[\sf (flattening-condition)] Assume that $\cR$ contains a rule of
  the form $R = (l\to r\Leftarrow\ol{s_n\tto t_n})$ where $s_i$ is
  neither a constructor term nor a basic term,
  $i\in\{1,\ldots,n\}$. Let $s_i|_q$, $q\in\pos(s_1)$, be a basic
  subterm of $s_i$. Then, we replace rule $R$ by a new rule of the
  form
  $l\to r \Leftarrow s_1\tto t_1,\ldots,s_i|_q \tto w,s_i[w]_q\tto
  t_i,\ldots,s_n\tto t_n$, where $w$ is a fresh variable.

\item[\sf (removal-unify)] Assume that $\cR$ contains a rule of the
  form $R = (l\to r\Leftarrow\ol{s_n\tto t_n})$ where $s_i$ is a
  constructor term, $i\in\{1,\ldots,n\}$. If
  $mgu(s_i,t_i) = \theta\neq \mathit{fail}$, then we replace rule $R$
  by a new rule of the form
  $l\theta \to r\theta \Leftarrow s_1\theta \tto t_1\theta,\ldots,
  s_{i-1}\theta \tto t_{i-1}\theta, s_{i+1}\theta \tto
  t_{i+1}\theta,\ldots,s_n\theta \tto t_n\theta$.

\item[\sf (removal-fail)] Assume that $\cR$ contains a rule of the
  form $R = (l\to r\Leftarrow\ol{s_n\tto t_n})$ where $s_i$ is a
  constructor term, $i\in\{1,\ldots,n\}$. If
  $mgu(s_i,t_i) = \mathit{fail}$, then we remove rule $R$ from $\cR$.
\end{description}
Trivially, by applying rule {\sf flattening-rhs} as much as possible,
we end up with a DCTRS where all the right-hand sides are constructor
terms; analogously, the exhaustive application of rule {\sf
  flattening-condition} allows us to ensure that the left-hand sides
of all equations in the conditions of the rules are either constructor
or basic; finally, the application of rules {\sf removal-unify} and
{\sf removal-fail} produces a \pc\ by removing those equations in
which the left-hand side is a constructor term.
Therefore, in the remainder of this paper, we only consider \pcs.

A nice property of \pcs\ is that one can consider reductions only
at \emph{topmost} positions.  Formally, given a \pc\ $\cR$, we say
that $s \to_{p,l\to r\Leftarrow \ol{s_n\tto t_n}} t$ is a \emph{top}
reduction step if $p=\epsilon$, there is a ground substitution
$\sigma$ with $s = l\sigma$, $s_i\sigma \to^\ast_\cR t_i\sigma$ for
all $i=1,\ldots,n$, $t = r\sigma$, and all the steps in $s_i\sigma
\to^\ast_\cR t_i\sigma$ for $i=1,\ldots,n$ are also top reduction
steps. We denote top reductions with $\topr$ for standard rewriting,
and $\stackrel{\epsilon}{\rh}_\cR,\stackrel{\epsilon}{\lh}_\cR$ for
our reversible rewrite relations.

The following result basically states that $\inn$ and $\topr$ are
equivalent for \pcs:

\begin{theorem} \label{th:basictop} Let $\cR$ be a constructor DCTRS
  with constructor terms in the right-hand sides of the equations and
  $\cR'$ be a \pc\ obtained from $\cR$ by a sequence of
  transformations of flattening and simplification of constructor
  conditions. Given ground terms $s$ and $t$ such that $s$ is basic
  and $t$ is normalized, we have $s \inns_\cR\: t$ iff
  $s \tops_{\cR'} \: t$.
\end{theorem}

\begin{proof}
  First, it is straightforward to see that an innermost reduction in
  $\cR'$ can only reduce the topmost positions of terms since defined
  functions can only occur at the root of terms and the terms
  introduced by instantiation are, by definition, irreducible.
  Therefore, the claim is a consequence of
  Theorems~\ref{th:flattening},
  \ref{thm:remove_constructor-conditions} and
  \ref{thm:remove_infeasible-rules}, together with the above
  fact. \qed
\end{proof}
Therefore, when considering \pcs\ and top reductions, storing the
reduced positions in the trace terms becomes redundant since they are
always $\epsilon$.  Thus, in practice, one can consider simpler trace
terms without positions, $\beta(\sigma,\hs_1,\ldots,\hs_n)$, that
implicitly represent the trace term
$\beta(\epsilon,\sigma,\hs_1,\ldots,\hs_n)$.

\begin{example} \label{ex:basic} Consider the following TRS $\cR$
  defining addition and multiplication on natural numbers, and its
  associated \pc\ $\cR'$:
  \[
  \begin{array}{l@{}r@{~}c@{~}l@{}l@{~}r@{~}c@{~}l}
    \cR = \{ & \mathsf{add}(\mathsf{0},y) & \to & y, \\
    & \mathsf{add}(\mathsf{s}(x),y) & \to &
    \mathsf{s}(\mathsf{add}(x,y)), \\
    & \mathsf{mult}(\mathsf{0},y) & \to & \mathsf{0}, \\
    & \mathsf{mult}(\mathsf{s}(x),y) & \to &
    \mathsf{add}(\mathsf{mult}(x,y),y) & \} 
  \end{array}
  \]
  \[
  \begin{array}{l@{}r@{~}c@{~}l@{}l@{~}r@{~}c@{~}l}
    \cR' = \{ & \mathsf{add}(\mathsf{0},y) & \to & y, \\
     & \mathsf{add}(\mathsf{s}(x),y) & \to &
    \mathsf{s}(z) \Leftarrow \mathsf{add}(x,y)\tto z,\\
     & \mathsf{mult}(\mathsf{0},y) & \to & \mathsf{0},\\
    & \mathsf{mult}(\mathsf{s}(x),y) & \to &
    w \Leftarrow \mathsf{mult}(x,y) \tto z, \mathsf{add}(z,y)\tto w & \}\\
  \end{array}
  \]
  For instance, given the following reduction in $\cR$:
  \[
  \mathsf{mult}(\mathsf{s}(\mathsf{0}),\mathsf{s}(\mathsf{0}))
  \inn_\cR
  \mathsf{add}(\mathsf{mult}(\mathsf{0},\mathsf{s}(\mathsf{0})),\mathsf{s}(\mathsf{0}))
  \inn_\cR 
  \mathsf{add}(\mathsf{0},\mathsf{s}(\mathsf{0}))
  \inn_\cR 
  \mathsf{s}(\mathsf{0})
  \]
  we have the following counterpart in $\cR'$:
  \[
  \begin{array}{lr@{~}l}
    \mathsf{mult}(\mathsf{s}(\mathsf{0}),\mathsf{s}(\mathsf{0}))
    \topr_{\cR'} \mathsf{s}(\mathsf{0})
    & \mbox{with} &
    \mathsf{mult}(\mathsf{0},\mathsf{s}(\mathsf{0}))\topr_{\cR'}
    \mathsf{0} \\
    & \mbox{and} & \mathsf{add}(\mathsf{0},\mathsf{s}(\mathsf{0}))
    \topr_{\cR'} \mathsf{s}(\mathsf{0}) 
  \end{array}
  \]
\end{example}
Trivially, all results in Section~\ref{sec:rr} hold for \pcs\ and
top reductions starting from basic terms. The simpler trace terms
without positions will allow us to introduce appropriate
injectivization and inversion transformations in the next section.

%%%%%%%%%%%%%%%%%%%%%%%%%%%%%%%%%%%%%%%%%%%%%%%%%%%%%%%%%%%%%%%%%%
\section{Reversibilization} \label{sec:transf}

In this section, we aim at \emph{compiling} the reversible extension
of rewriting into the system rules. Intuitively speaking, given a pure
constructor system $\cR$, we aim at producing new systems $\cR_f$ and
$\cR_b$ such that standard rewriting in $\cR_f$, i.e., $\to_{\cR_f}$,
coincides with the forward reversible extension $\rh_\cR$ in the
original system, and analogously $\to_{\cR_b}$ is equivalent to
$\lh_\cR$. Therefore, $\cR_f$ can be seen as an \emph{injectivization}
of $\cR$, and $\cR_b$ as the \emph{inversion} of $\cR_f$.

In principle, we could easily introduce a transformation for \pcs\
that mimicks the behavior of the reversible extension of rewriting.
For instance, given the \pc\ $\cR$ of
Example~\ref{ex:needforvbles}, we could produce the following
injectivized version $\cR_f$:\footnote{We will write just $\beta$
  instead of $\beta()$ when no argument is required.}
\[
\begin{array}{@{}r@{~}c@{~}l@{~}l}
  \tuple{\mathsf{f}(x,y,m),ws} & \to &
  \tuple{\mathsf{s}(w),\beta_1(m,x,w_1,w_2):ws}  \\
  && \mbox{}\hspace{2ex}\Leftarrow  \tuple{\mathsf{h}(x),\nil} \tto
  \tuple{x,w_1}, \tuple{\mathsf{g}(y,\mathsf{4}),\nil}\tto \tuple{w,w_2}\\
  \tuple{\mathsf{h}(\mathsf{0}),ws} &  \to &
  \tuple{\mathsf{0},\beta_2:ws} \\
  \tuple{\mathsf{h}(\mathsf{1}),ws} & \to & \tuple{\mathsf{1},\beta_3:ws} \\
  \tuple{\mathsf{g}(x,y),ws} & \to & \tuple{x,\beta_4(y):ws} \\
\end{array}
\]
For instance, the reversible step
$\tuple{\mathsf{f}(\mathsf{0},\mathsf{2},\mathsf{4}),\nil}
\stackrel{\epsilon}{\rh}_\cR
\tuple{\mathsf{s}(\mathsf{2}),[\beta_1(\sigma',\hs_1,\hs_2)]}$
with $\sigma' = \{m\mapsto\mathsf{4},x\mapsto\mathsf{0}\}$, $\hs_1 =
[\beta_2(\id)]$ and $\hs_2 =
[\beta_4(\{y\mapsto\textsf{4}\})]$, has the following
counterpart in $\cR_f$:
\[
\begin{array}{lll}
  \tuple{\mathsf{f}(\mathsf{0},\mathsf{2},\mathsf{4}),\nil} \topr_{\cR_f}
  \tuple{\mathsf{s}(\mathsf{2}),[\beta_1(\mathsf{4},\mathsf{0},[\beta_2],[\beta_4(\mathsf{4})])]} \\
  \mbox{}\hspace{10ex}\mbox{with} ~~\tuple{\mathsf{h}(\mathsf{0}),\nil}\topr_{\cR_f}
  \tuple{\mathsf{0},[\beta_2]}~~\mbox{and}~~\tuple{\mathsf{g}(\mathsf{2},\mathsf{4}),\nil} \topr_{\cR_f} \tuple{\mathsf{2},[\beta_4(\mathsf{4})]}\\
\end{array}
\]
The only subtle difference here is that a trace term like
\[
\beta_1(\{m\mapsto\mathsf{4},x\mapsto\mathsf{0}\},[\beta_2(\id)],[\beta_4(\{y\mapsto\textsf{4}\})])
\]
is now stored in the transformed system as
\[
\beta_1(\mathsf{4},\mathsf{0},[\beta_2],[\beta_4(\mathsf{4})])
\]
Furthermore, we could produce an inverse $\cR_b$ of the above system
as follows:
\[
\begin{array}{r@{~}c@{~}l@{~}l}
  \tuple{\mathsf{s}(w),\beta_1(m,x,w_1,w_2):ws}^{-1} & \to &
  \tuple{\mathsf{f}(x,y,m),ws}^{-1} \\
  && \mbox{}\hspace{2ex}\Leftarrow \tuple{w,w_2}^{-1} \tto \tuple{\mathsf{g}(y,\mathsf{4}),\nil}^{-1},\\
  && \mbox{}\hspace{5.5ex}\tuple{x,w_1}^{-1} \tto \tuple{\mathsf{h}(x),\nil}^{-1} \\
  \tuple{\mathsf{0},\beta_2:ws}^{-1} &  \to &
  \tuple{\mathsf{h}(\mathsf{0}),ws}^{-1} \\
  \tuple{\mathsf{1},\beta_3:ws}^{-1} & \to & \tuple{\mathsf{h}(\mathsf{1}),ws}^{-1} \\
  \tuple{x,\beta_4(y):ws}^{-1} & \to & \tuple{\mathsf{g}(x,y),ws}^{-1} \\
\end{array}
\]
mainly by switching the left- and right-hand sides of each rule and
condition. The correctness of these injectivization and inversion
transformations would be straightforward.

These transformations are only aimed at mimicking, step by step, the
reversible relations $\rh_{\cR}$ and $\lh_{\cR}$. Roughly speaking,
for each step $\tuple{s,\hs} \rh_{\cR} \tuple{t,\hs'}$ in a system
$\cR$, we have $\tuple{s,\hs} \to_{\cR_f} \tuple{t,\hs'}$, where
$\cR_f$ is the injectivized version of $\cR$, and for each step
$\tuple{s,\hs} \lh_{\cR} \tuple{t,\hs'}$ in $\cR$, we have
$\tuple{s,\hs} \to_{\cR_b} \tuple{t,\hs'}$, where $\cR_b$ is the
inverse of $\cR_f$. More details on this approach can be found in
\cite{NPV16}. Unfortunately, it might be much more useful to produce
injective and inverse versions of \emph{each function} defined in a
system $\cR$. Note that, in the above approach, the system $\cR_f$
only defines a single function $\tuple{\_,\_}$ and $\cR_b$ only
defines $\tuple{\_,\_}^{-1}$, i.e., we are computing systems that
define the relations $\rh_\cR$ and $\lh_\cR$ rather than the
injectivized and inverse versions of the functions in $\cR$.
In the following, we introduce more refined transformations that can
actually produce injective and inverse versions of the original
functions.

\subsection{Injectivization} \label{sec:injectivization}

In principle, given a function $\mathsf{f}$, one can consider that the
injectivization of a rule of the form\footnote{By abuse of notation,
  here we let $\ol{s_0},\ldots,\ol{s_n}$ denote sequences of terms of
  arbitrary length, i.e., $\ol{s_{0}} = s_{0,1},\ldots,s_{0,l_0}$,
  $\ol{s_1} = s_{1,1},\ldots,s_{1,l_1}$, etc.}
\[
\beta:\mathsf{f}(\ol{s_0}) \to r \Leftarrow
\mathsf{f_1}(\ol{s_{1}})\tto t_1,\ldots,\mathsf{f_n}(\ol{s_{n}})\tto
t_n
\]
produces the following rule
\[
\mathsf{f}^\mathtt{i}(\ol{s_0}) \to \tuple{r,\beta(\ol{y},\ol{w_n})}
\Leftarrow \mathsf{f}_1^\mathtt{i}(\ol{s_{1}})\tto
\tuple{t_1,w_1}\ldots,\mathsf{f}_n^\mathtt{i}(\ol{s_{n}})\tto
\tuple{t_n,w_n}
\]
where $\{\ol{y}\} = (\var(l)\backslash\var(r,\ol{s_n},\ol{t_n}))\cup
\bigcup_{i=1}^n \var(t_i)\backslash\var(r,\ol{s_\interval{i+1}{n}})$
and $\ol{w_n}$ are fresh variables.
The following example, though, illustrates that this is not
correct in general.

\begin{example}
  Consider the following \pc\ $\cR$:
  \[
  \begin{array}{lrcl}
    \beta_1: & \mathsf{f}(x,y) & \to & z \Leftarrow \mathsf{h}(y)\tto
    w,~\mathsf{first}(x,w)\tto z\\
    \beta_2: & \mathsf{h}(\mathsf{0}) & \to & \mathsf{0} \\
    \beta_3: & \mathsf{first}(x,y) & \to & x \\
  \end{array}
  \]
  together with the following top reduction: 
  \[
  \begin{array}{@{}ll@{}}
  \mathsf{f}(\mathsf{2},\mathsf{1}) \topr_\cR \mathsf{2} & \mbox{with}~\sigma=\{x\mapsto
  \mathsf{2},y\mapsto \mathsf{1},w\mapsto
  \mathsf{h}(\mathsf{1}),z\mapsto\mathsf{2}\}\\
  & \mbox{where}~\mathsf{h}(y)\sigma = \mathsf{h}(\mathsf{1})\tops_\cR 
  \:\mathsf{h}(\mathsf{1})=w\sigma\\
  & \mbox{and}~
  \mathsf{first}(x,w)\sigma=\mathsf{first}(\mathsf{2},\mathsf{h}(\mathsf{1}))
  \topr_\cR \mathsf{2}=z\sigma
  \end{array}
  \]
  Following the scheme above, we would produce the following \pc\:
  \[
  \begin{array}{r@{~}c@{~}l}
    \mathsf{f}^\mathtt{i}(x,y) & \to & \tuple{z,\beta_1(w_1,w_2)} \Leftarrow \mathsf{h}^\mathtt{i}(y)\tto
    \tuple{w,w_1},~\mathsf{first}^\mathtt{i}(x,w)\tto \tuple{z,w_2}\\
    \mathsf{h}^\mathtt{i}(\mathsf{0}) & \to & \tuple{\mathsf{0},\beta_2} \\
    \mathsf{first}^\mathtt{i}(x,y) & \to & \tuple{x,\beta_3(y)} \\
  \end{array}
  \]
  Unfortunately, the corresponding reduction for
  $\mathsf{f}^\mathtt{i}(\mathsf{2},\mathsf{1})$ above cannot be done in this
  system since $\textsf{h}^\mathtt{i}(\textsf{1})$ cannot be reduced to
  $\tuple{\textsf{h}^\mathtt{i}(\textsf{1}),\nil}$.
\end{example}
In order to overcome this drawback, one could \emph{complete} the
function definitions with rules that reduce each irreducible term $t$
to a tuple of the form $\tuple{t,\nil}$. Although we find it a
promising idea for future work, in this paper we propose a simpler
approach. In the following, we consider a refinement of innermost
reduction where only constructor substitutions are computed. Formally,
the constructor reduction relation, $\cinn$, is defined as follows:
given ground terms $s,t\in\cT(\cF)$, we have $s \cinn_\cR t$ iff there
exist a position $p$ in $s$ such that no proper subterms of $s|_p$ are
reducible, a rewrite rule $l \to r\Leftarrow \ol{s_n\tto t_n} \in
\cR$, and a ground \emph{constructor} substitution $\sigma$ such that
$s|_p = l\sigma$, $s_i\sigma \cinns_\cR \: t_i\sigma$ for all
$i=1,\ldots,n$, and $t = s[r\sigma]_p$.
Note that the results in the previous section also hold for $\cinn$.

In the following, given a basic term $t = \mathsf{f}(\ol{s})$, we
denote by $t^\mathtt{i}$ the term $\mathsf{f}^\mathtt{i}(\ol{s})$.  Now, we introduce
our injectivization transformation as follows:

\begin{definition}[injectivization]
  Let $\cR$ be a \pc. We produce a new CTRS $\mathbf{I}(\cR)$ by
  replacing each rule $ \beta: l \to r \Leftarrow \ol{s_n\tto t_n} $
  of $\cR$ by a new rule of the form
  \[
  l^\mathtt{i} \to \tuple{r,\beta(\ol{y},\ol{w_n})} \Leftarrow \ol{s_n^\mathtt{i} \tto \tuple{t_n,w_n}}
  \]
  in $\mathbf{I}(\cR)$, where $\{\ol{y}\} =
  (\var(l)\backslash\var(r,\ol{s_n},\ol{t_n}))\cup \bigcup_{i=1}^n
  \var(t_i)\backslash\var(r,\ol{s_\interval{i+1}{n}})$
  and $\ol{w_n}$ are fresh variables.  Here, we assume that the
  variables of $\ol{y}$ are in lexicographic order.
\end{definition}
Observe that now we do not need to keep a trace in each term, but only
a single trace term since all reductions finish in one step in a
\pc. The relation between the original trace terms and the
information stored in the injectivized system is formalized as
follows:

\begin{definition}
  Given a trace term $\hs=\beta(\{\ol{y_m\mapsto
    t_m}\},\hs_1,\ldots,\hs_n)$, we define $\widehat{\hs}$ recursively
  as follows: $\widehat{\hs} =
  \beta(\ol{t_m},\widehat{\hs_1},\ldots,\widehat{\hs_n})$, where we
  assume that the variables $\ol{y_m}$ are in lexicographic order.
\end{definition}
Moreover, in order to simplify the notation, we consider that a a
trace term $\hs$ and a singleton list of the form $[\hs]$ denote the
same object.
The correctness of the injectivization transformation is stated as
follows:

\begin{theorem} \label{th:injectivization2} Let $\cR$ be a 
  \pc\ and $\cR_f=\mathbf{I}(\cR)$ be its injectivization. Then
  $\cR_f$ is a \pc\ and, given a basic ground term $s$, we
  have $\tuple{s,\nil} \stackrel{\mathsf{c}}{\rh}_\cR \tuple{t,\hs}$
  iff $s^\mathtt{i} \cinn_{\cR_f} \tuple{t,\widehat{\hs}}$.
\end{theorem}

\begin{proof}
  The fact that $\cR_f$ is a \pc\ is trivial.  Regarding the
  second part, we proceed as follows:

  $(\Rightarrow)$ We proceed by induction on the depth $k$ of the step
  $\tuple{s,\nil} \stackrel{\mathsf{c}}{\rh}_{\cR_k} \tuple{t,\hs}$.
  Since the depth $k=0$ is trivial, we consider the inductive case
  $k>0$. Thus, there is a rule $\beta:l\to r \Leftarrow \ol{s_n\tto
    t_n}\in\cR$, and a substitution $\sigma$ such that $s=l\sigma$,
  $\tuple{s_i\sigma,\nil}
  {\stackrel{\mathsf{c}}{\rh}\!\!{}_{\cR_{k_i}}}
  \tuple{t_i\sigma,\hs_{i}}$, $i=1,\ldots,n$, $t=r\sigma$,
  $\sigma'=\sigma\!\!\res_{(\var(l)\backslash\var(r,\ol{s_n},\ol{t_n}))\cup
    \bigcup_{i=1}^n
    \var(t_i)\backslash\var(r,\ol{s_\interval{i+1}{n}})}$, and
  $\hs=\beta(\sigma',\hs_{1},\ldots,\hs_{n})$.  By definition of
  $\rh_{\cR_k}$, we have that $k_i<k$ for all $i=1,\ldots,n$ and,
  thus, by the induction hypothesis, we have $(s_i\sigma)^\mathtt{i}
  \cinn_{\cR_{f}} \: 
  \tuple{t_i\sigma,\widehat{\hs}_{i}}$
  for all $i=1,\ldots,n$. Consider now the equivalent rule in $\cR_f$:
  $l^\mathtt{i} \to \tuple{r,\beta(\ol{y},\ol{w_n})} \Leftarrow s_1^\mathtt{i} \tto
  \tuple{t_1,w_1},\ldots,s_n^\mathtt{i} \tto \tuple{t_n,w_n}$.  Therefore, we have
  $s^\mathtt{i} \cinn_{\cR_{f}}
  \tuple{t,\beta(\ol{y}\sigma,\widehat{\hs}_{1},\ldots,\widehat{\hs}_{n})}$
  where $\{\ol{y}\} = (\var(l)\backslash\var(r,\ol{s_n},\ol{t_n}))\cup
  \bigcup_{i=1}^n
  \var(t_i)\backslash\var(r,\ol{s_\interval{i+1}{n}})$ and, thus, we
  can conclude that $\widehat{\hs} =
  \beta(\ol{y}\sigma,\widehat{\hs}_{1},\ldots,\widehat{\hs}_{n})$.

  $(\Leftarrow)$ This direction is analogous. We proceed by induction
  on the depth $k$ of the step $s^\mathtt{i} \cinn_{\cR_{f_k}}\:
  \tuple{t,\widehat{\hs}}$.  Since the depth $k=0$ is trivial, we
  consider the inductive case $k>0$. Thus, there is a rule $l^\mathtt{i} \to
  \tuple{r,\beta(\ol{y},\ol{w_n})} \Leftarrow s_1^\mathtt{i} \tto
  \tuple{t_1,w_1},\ldots,s_n^\mathtt{i} \tto \tuple{t_n,w_n}$ in $\cR_f$ and a
  substitution $\theta$ such that $l^\mathtt{i}\theta = s^\mathtt{i}$, $s^\mathtt{i}_i\theta
  \cinn_{\cR_{f_{k_i}}} \tuple{t_i,w_i}\theta$, $i = 1,\ldots,n$, and
  $\tuple{r,\beta(\ol{y},\ol{w_n})}\theta = \tuple{t,\widehat{\hs}}$.
  Assume that $\sigma$ is the restriction of $\theta$ to the variables
  of the rule, excluding the fresh variables $\ol{w_n}$, and that
  $w_i\theta = \widehat{\hs}_{i}$ for all $i=1,\ldots,n$. Therefore,
  $\tuple{s_i,\nil}\theta = \tuple{s_i\sigma,\nil}$ and
  $\tuple{t_i,w_i}\theta = \tuple{t_i\sigma,\widehat{\hs}_{i}}$,
  $i=1,\ldots,n$. Then, by definition of $\cR_{f_{k_i}}$, we have that
  $k_i<k$ for all $i=1,\ldots,n$ and, thus, by the induction
  hypothesis, we have $\tuple{s_i\sigma,\nil}
  {\stackrel{\mathsf{c}}{\rh}\!\!{}_{\cR}}
  \tuple{t_i\sigma,\hs_{i}}$, $i=1,\ldots,n$. Consider now the
  equivalent rule in $\cR$: $\beta:l\to r \Leftarrow \ol{s_n\tto
    t_n}\in\cR$. Therefore, we have $\tuple{s,\nil}
  \stackrel{\mathsf{c}}{\rh}_\cR \tuple{t,\hs}$,
  $\sigma'=\sigma\!\!\res_{(\var(l)\backslash\var(r,\ol{s_n},\ol{t_n}))\cup
    \bigcup_{i=1}^n
    \var(t_i)\backslash\var(r,\ol{s_\interval{i+1}{n}})}$, and
  $\hs=\beta(\sigma',\hs_{1},\ldots,\hs_{n})$. Finally, since
  $\{\ol{y}\} = (\var(l)\backslash\var(r,\ol{s_n},\ol{t_n}))\cup
  \bigcup_{i=1}^n
  \var(t_i)\backslash\var(r,\ol{s_\interval{i+1}{n}})$, we can
  conclude that $\widehat{\hs} = \hs$.  \qed
\end{proof}

\subsection{Inversion}

Given an injectivized system, inversion basically amounts to switching
the left- and right-hand sides of the rule and of every equation in
the condition, as follows:

\begin{definition}[inversion] \label{def:inversion2} Let $\cR$ be a
  \pc\ and $\cR_f = \mathbf{I}(\cR)$ be its injectivization. The
  inverse system $\cR_b = \mathbf{I}^{-1}(\cR_f)$ is obtained from
  $\cR_f$ by replacing each rule\footnote{Here, we assume that
    $\ol{s_0}$, $\ol{s_1}$,\ldots, $\ol{s_n}$ denote arbitrary
    sequences of terms, i.e., $\ol{s_{0}} = s_{0,1},\ldots,s_{0,l_0}$,
    $\ol{s_1} = s_{1,1},\ldots,s_{1,l_1}$, etc. We use this notation
    for clarity.}
  \[
  \mathsf{f}^\mathtt{i}(\ol{s_0}) \to \tuple{r,\beta(\ol{y},\ol{w_n})} \Leftarrow
  \mathsf{f}_1^\mathtt{i}(\ol{s_{1}})\tto
  \tuple{t_1,w_1},\ldots,\mathsf{f}_n^\mathtt{i}(\ol{s_{n}})\tto \tuple{t_n,w_n}
  \]
  of $\cR_f$ by a new rule of the form
  \[
  \mathsf{f}^{-1}(r,\beta(\ol{y},\ol{w_n})) \to \tuple{\ol{s_0}} \Leftarrow
  \mathsf{f}_n^{-1}(t_n,w_n)\tto
  \tuple{\ol{s_{n}}},\ldots,\mathsf{f}_1^{-1}(t_1,w_1)\tto \tuple{\ol{s_{1}} }
  \]
  in $\mathbf{I}^{-1}(\cR_f)$, where the variables of $\ol{y}$ are in
  lexicographic order.
\end{definition}

\begin{example}
  Consider again the \pc\ of Example~\ref{ex:needforvbles}. Here,
  injectivization returns the following \pc\ $\mathbf{I}(\cR) =
  \cR_f$:
  \[
  \begin{array}{r@{~}c@{~}l}
    \mathsf{f}^\mathtt{i}(x,y,m) & \to &
      \tuple{\mathsf{s}(w),\beta_1(m,x,w_1,w_2)}  \\
      &&\mbox{}\hspace{2ex}\Leftarrow \mathsf{h}^\mathtt{i}(x) \tto
      \tuple{x,w_1}, \mathsf{g}^\mathtt{i}(y,\mathsf{4})\tto \tuple{w,w_2}\\
      \mathsf{h}^\mathtt{i}(\mathsf{0}) &  \to &
      \tuple{\mathsf{0},\beta_2} \\
      \mathsf{h}^\mathtt{i}(\mathsf{1}) & \to & \tuple{\mathsf{1},\beta_3} \\
      \mathsf{g}^\mathtt{i}(x,y) & \to & \tuple{x,\beta_4(y)} \\
  \end{array}
  \]
  Then, inversion with $\mathbf{I}^{-1}$ produces the following \pc\
  $\mathbf{I}^{-1}(\mathbf{I}(\cR)) = \cR_b$:
  \[
  \begin{array}{r@{~}c@{~}l}
    \mathsf{f}^{-1}(\mathsf{s}(w),\beta_1(m,x,w_1,w_2)) & \to &
    \tuple{x,y,m} \\
      &&\mbox{}\hspace{2ex}\Leftarrow \mathsf{g}^{-1}(w,w_2)\tto \tuple{y,\mathsf{4}}, \mathsf{h}^{-1}(x,w_1) \tto \tuple{x}\\
    \mathsf{h}^{-1}(\mathsf{0},\beta_2) &  \to &
    \tuple{\mathsf{0}} \\
    \mathsf{h}^{-1}(\mathsf{1},\beta_3) & \to & \tuple{\mathsf{1}} \\
    \mathsf{g}^{-1}(x,\beta_4(y)) & \to & \tuple{x,y } \\
  \end{array}
  \]
\end{example}
Finally, the correctness of the inversion transformation is stated as
follows:

\begin{theorem} \label{th:inversion2} Let $\cR$ be a \pc,
  $\cR_f=\mathbf{I}(\cR)$ its injectivization, and $\cR_b=
  \mathbf{I}^{-1}(\cR_f)$ the inversion of $\cR_f$. Then, $\cR_b$ is a
  basic \pc\ and, given a basic ground term $\mathsf{f}(\ol{s})$
  and a constructor ground term $t$ with $\tuple{t,\hs}$ a safe pair,
  we have $\tuple{t,\hs} \stackrel{\mathsf{c}}{\lh}_\cR
  \tuple{\mathsf{f}(\ol{s}),\nil}$ iff
  $\mathsf{f}^{-1}(t,\widehat{\hs}) \cinn_{\cR_b}\: \tuple{\ol{s}}$.
\end{theorem}

\begin{proof}
  The fact that $\cR_f$ is a \pc\ is trivial.  Regarding the
  second part, we proceed as follows.

  $(\Rightarrow)$ We proceed by induction on the depth $k$ of the step
  $\tuple{t,\hs} \stackrel{\mathsf{c}}{\lh}_{\cR_k}
  \tuple{\mathsf{f}(\ol{s}),\nil}$.  Since the depth $k=0$ is trivial,
  we consider the inductive case $k>0$. Let
  $\hs=\beta(\sigma',\ol{\hs_n})$.  Thus, we have that
  $\tuple{t,\beta(\sigma',\ol{\hs_n})}$ is a safe pair, there is a
  rule $\beta: \mathsf{f}(\ol{s_0}) \to r \Leftarrow
  \mathsf{f}_1(\ol{s_{1}})\tto t_1,\ldots,\mathsf{f}_n(\ol{s_{n}})\tto
  t_n$ and substitution $\theta$ with
  $\dom(\theta)=(\var(r,\ol{s_1},\ldots,\ol{s_n})\backslash\dom(\sigma'))$
  such that $t = r\theta$, $\tuple{t_i\theta\sigma',\hs_i}
  {\cinn_{\cR_{k_i}}} \tuple{\mathsf{f}(\ol{s_i})\theta\sigma',\nil}$
  for all $i=1,\ldots,n$, and $\mathsf{f}(\ol{s}) =
  \mathsf{f}(\ol{s_0})\theta\sigma'$.  
  Note that $\ol{s_0},\ldots,\ol{s_n}$ denote sequences of terms of
  arbitrary length, i.e., $\ol{s_{0}} = s_{0,1},\ldots,s_{0,l_0}$,
  $\ol{s_1} = s_{1,1},\ldots,s_{1,l_1}$, etc.
  Since $\tuple{t,\hs}$ is a safe pair, we have that $\dom(\sigma')=
  (\var(\ol{s_0})\backslash\var(r,\ol{s_1},\ldots,\ol{s_n},\ol{t_n}))\cup
  \bigcup_{i=1}^n
  \var(t_i)\backslash\var(r,\ol{s_{i+1}},\ldots,\ol{s_{n}})$.
  By definition of $\lh_{\cR_k}$, we have that $k_i<k$ for all
  $i=1,\ldots,n$ and, by the induction hypothesis, we have
  $\mathsf{f}^{-1}(t_i\sigma,\widehat{\hs_i}) {\cinn_{\cR_{b}}}
  \tuple{\ol{s_i}\sigma} $ for all $i=1,\ldots,n$. Let us now consider
  the equivalent rule in $\cR_b$:
  \[
  \mathsf{f}^{-1}(r,\beta(\ol{y},\ol{w_n}))) \to \tuple{\ol{s_0}}
  \Leftarrow \mathsf{f}_n^{-1}(t_n,w_n) \tto \tuple{\ol{s_{n}}},
  \ldots, \mathsf{f}_1^{-1}(t_1,w_1) \tto \tuple{\ol{s_{1}}}
  \]
  Hence, we have
  $\mathsf{f}^{-1}(t,\beta(\ol{y}\sigma,\widehat{\hs}_{1},\ldots,\widehat{\hs}_{1}))
  \to_{\cR_{b}} \tuple{\ol{s_0}\sigma} = \tuple{\ol{s}}$, where
  \[
  \{\ol{y}\} =
  (\var(\ol{s_0})\backslash\var(r,\ol{s_1},\ldots,\ol{s_n},\ol{t_n}))\cup
  \bigcup_{i=1}^n \var(t_i)\backslash\var(r,\ol{s_{i+1}},\ldots,\ol{s_{n}})
  \]
  and, thus, we can conclude that $\widehat{\hs} =
  \beta(\ol{y}\sigma,\widehat{\hs}_{1},\ldots,\widehat{\hs}_{n})$.

  $(\Leftarrow)$ This direction is analogous. We proceed by induction
  on the depth $k$ of the step $\mathsf{f}^{-1}(t,\widehat{\hs})
  \cinn_{\cR_{b_k}} \tuple{\ol{s}}$.  Since the depth $k=0$ is
  trivial, we consider the inductive case $k>0$. Thus, there is a rule
  $\mathsf{f}^{-1}(r,\beta(\ol{y},\ol{w_n}))) \to \tuple{\ol{s_0}}
  \Leftarrow \mathsf{f}_n^{-1}(t_n,w_n) \tto
  \tuple{\ol{s_{n}}},\ldots,\mathsf{f}_1^{-1}(t_1,w_1) \tto
  \tuple{\ol{s_{1}}}$ in $\cR_b$ and a substitution $\theta$ such that
  $\mathsf{f}^{-1}(r,\beta(\ol{y},\ol{w_n}))\theta =
  \mathsf{f}^{-1}(t,\widehat{\hs})$, $\mathsf{f}_i^{-1}(t_i,w_i)\theta
  \cinn_{\cR_{b_{k_i}}} \tuple{\ol{s_i}}\theta$, $i = n,\ldots,1$, and
  $\mathsf{f}^{-1}(r,ws)\theta = \tuple{\ol{s}}$.  Assume that
  $\sigma$ is the restriction of $\theta$ to the variables of the
  rule, excluding the fresh variables $\ol{w_n}$, and that $w_i\theta
  = \widehat{\hs}_{i}$ for all $i=1,\ldots,n$. Therefore,
  $\mathsf{f}^{-1}(r,\beta(\ol{y},\ol{w_n}))\theta =
  \mathsf{f}^{-1}(r\sigma,
  \beta(\ol{y}\sigma,\widehat{\hs}_{1},\ldots,\widehat{\hs}_{n})$,
  $\mathsf{f}_i^{-1}(t_i,w_i)\theta =
  \mathsf{f}_i^{-1}(t_i\sigma,\widehat{\hs}_{i})$ and
  $\tuple{\ol{s_i}}\theta = \tuple{\ol{s_i}\sigma}$,
  $i=1,\ldots,n$. Then, by definition of $\cR_{b_{k_i}}$, we have that
  $k_i<k$ for all $i=1,\ldots,n$ and, thus, by the induction
  hypothesis, we have $\tuple{t_i\sigma,\hs_{i}}
  {\stackrel{\mathsf{c}}{\lh}\!\!{}_{\cR}}
  \tuple{\mathsf{f}_i(\ol{s_i}\sigma),\nil}$, $i=1,\ldots,n$.
  Consider now the equivalent rule in $\cR$: $\beta:
  \mathsf{f}(\ol{s_0}) \to r \Leftarrow \mathsf{f}_1(\ol{s_{1}})\tto
  t_1,\ldots,\mathsf{f}_n(\ol{s_{n}})\tto t_n$ in $\cR$. Therefore, we
  have $\tuple{t,\hs} \stackrel{\mathsf{c}}{\lh}_\cR
  \tuple{\mathsf{f}(\ol{s}),\nil}$,
  \[
  \sigma'=\sigma\!\!\res_{(\var(\ol{s_0})\backslash\var(r,\ol{s_1},\ldots,\ol{s_n},\ol{t_n}))\cup
    \bigcup_{i=1}^n
    \var(t_i)\backslash\var(r,\ol{s_{i+1}},\ldots,\ol{s_{n}})}
  \]
  and $\hs=\beta(\sigma',\hs_{1},\ldots,\hs_{n})$.  Finally, since
  $\{\ol{y}\} =
  (\var(\ol{s_0})\backslash\var(r,\ol{s_1},\ldots,\ol{s_n},\ol{t_n}))\cup
  \bigcup_{i=1}^n
  \var(t_i)\backslash\var(r,\ol{s_{i+1}},\ldots,\ol{s_n})$, we can
  conclude that $\widehat{\hs} = \hs$.  \qed
\end{proof}

%%%%%%%%%%%%%%%%%%%%%%%%%%%%%%%%%%%%%%%%%%%%%%%%%%%%%%%%%%%%%%%%%%
\subsection{Improving the transformation for injective functions} \label{sec:improved}

When a function is injective, one can expect the injectivization
transformation to be unnecessary. This is not generally true, since
some additional syntactic conditions might also be
required. Furthermore, depending on the considered setting, it can be
necessary to have an injective \emph{system}, rather than an injective
function. Consider, e.g., the following simple TRS:
\[
\cR = \{~\mathsf{f_1} \to \mathsf{f_2}, \mathsf{f_2} \to
\mathsf{0},~\mathsf{g_1} \to \mathsf{g_2},~\mathsf{g_2} \to
\mathsf{0}~\}
\]
Here, all functions are clearly injective. However, given a reduction
like $\mathsf{f_1} \to_\cR \mathsf{f_2} \to_\cR \mathsf{0}$, we do not
know which rule should be applied to $\mathsf{0}$ in order to go
backwards until the initial term (actually, both the second and the
fourth rules are applicable in the reverse direction).

Luckily, in our context, the injectivity of a function suffices since
reductions in \pcs\ are performed in a single step. Therefore, given a
reduction of the form $\mathsf{f}^\mathtt{i}(\ol{s_n}) \to_\cR t$, a
backward computation will have the form $\mathsf{f}^{-1}(t) \to_\cR
\tuple{\ol{s_n}}$, so that we know that only the inverse rules of
$\mathsf{f}$ are applicable.

Now, we present an improvement of the injectivization transformation
presented in Section~\ref{sec:injectivization} which has some
similarities with that in \cite{MHNHT07}. Here, we consider that the
initial system is a TRS $\cR$ since, to the best of our
knowledge, there is no reachability analysis defined for DCTRSs. In
the following, given a term $s$, we let
\[
\range(s) = \{ t \mid s\sigma \to_\cR^\ast
t,~\sigma:\cV\mapsto\cT(\cC),~\mbox{and}~ t\in\cT(\cC)\}
\]
i.e., $\range(s)$ returns a set with the constructor normal forms of
all possible ground constructor instances of $s$. Although computing
this set is generally undecidable, there are some overapproximations
based on the use of tree automata (see, e.g., \cite{Gen98} and the
most recent approach for innermost rewriting \cite{GS15}). Let us
consider that $\range^\alpha(s)$ is such an approximation, with
$\range^\alpha(s)\supseteq\range(s)$ for all terms $s$. Here, we are
interested in determining when the right-hand sides, $r_1$ and $r_2$,
of two rules do not overlap, i.e., $\range(r_1)\cap\range(r_2) =
\emptyset$.  For this purpose, we will check whether
$\range^\alpha(r_1)\cap\range^\alpha(r_2) = \emptyset$.  Since finite
tree automata are closed under intersection and the emptiness of a
finite tree automata is decidable, %%% TATA
checking the emptiness of $\range^\alpha(r_1)\cap\range^\alpha(r_2)$
is decidable and can be used to \emph{safely} identify non-overlapping
right-hand sides, i.e., if 
$\range^\alpha(r_1)\cap\range^\alpha(r_2) = \emptyset$, then $r_1$ and
$r_2$ are definitely non-overlapping; otherwise, they may be
overlapping or non-overlapping.

Now, we summarize our method to simplify some trace terms. Given a
constructor TRS $\cR$ and a rule $\beta: l\to r\in\cR$, we check the
following conditions:
\begin{enumerate}
\item the right-hand side $r$ of the rule does not overlap with the
  right-hand side of any other rule defining the same function;
\item the rule is non-erasing, i.e., $\var(l) = \var(r)$;
\item the right-hand side $r$ contains a single occurrence of a
  defined function symbol, say $\mathsf{f} \in\cD$.
\end{enumerate}
If these conditions hold, then the rule has the form
$l \to r[\mathsf{f}(\ol{s})]_p$ with $l$ and $\mathsf{f}(\ol{s})$
basic terms,\footnote{Note that $l$ is a basic term since we initially
  consider a constructor TRS and, thus, all left-hand sides are basic
  terms by definition.} and $r[x]_p$ and $\ol{s}$ constructor terms,
where $x$ is a fresh variable. In this case, we can safely produce the
following injective version:\footnote{Since $l\to r$ is non-erasing,
  the \pc\ rule $l \to r[x]_p \Leftarrow \mathsf{f}(\ol{s})\tto x$ is
  trivially non-erasing too (according to \cite{NSS12lmcs}, i.e.,
  $(\var(l)\backslash\var(r[x]_p,\mathsf{f}(\ol{s}),x))\cup
  \var(x)\backslash\var(r[x]_p) = \emptyset$) and, thus, no binding
  should be stored during the injectivization process.}
\[
l^\mathtt{i} \to \tuple{r[x]_p,w} \Leftarrow
\mathsf{f}^\mathtt{i}(\ol{s})\tto\tuple{x,w}
\]
instead of 
\[
l^\mathtt{i} \to \tuple{r[x]_p,\beta(w)} \Leftarrow
\mathsf{f}^\mathtt{i}(\ol{s})\tto\tuple{x,w}
\]
Let us illustrate this improved transformation with a couple of
examples.

\begin{example}
  Consider the following TRS:
  \[
  \cR = \{~\mathsf{f}(\mathsf{s}(x))\to\mathsf{g}(x),
  ~\mathsf{f}(\mathsf{c}(x))\to\mathsf{h}(x),
  ~\mathsf{g}(x)\to\mathsf{s}(x), ~\mathsf{h}(x)\to\mathsf{c}(x) \}
  \]
  Here, it can easily be shown that
  $\range^\alpha(\mathsf{g}(x))\cap\range^\alpha(\mathsf{h}(x)) =
  \emptyset$, the two rules defining $\mathsf{f}$ are non-erasing, and
  both contain a single occurrence of a defined function symbol in the
  righ-hand sides. Therefore, our improved injectivization applies and
  we get the following \pc\ $\cR_f$:
  \[
  \begin{array}{r@{~}c@{~}l@{~~~~~~~~~}r@{~}c@{~}l}
    \mathsf{f}^\mathtt{i}(\mathsf{s}(x)) & \to &  \tuple{y,w} \Leftarrow
    \mathsf{g}^\mathtt{i}(x)\tto \tuple{y,w}  
    &
    \mathsf{g}^\mathtt{i}(x) & \to & \tuple{\mathsf{s}(x),\beta_3}\\
    \mathsf{f}^\mathtt{i}(\mathsf{c}(x)) & \to & \tuple{y,w} \Leftarrow \mathsf{h}^\mathtt{i}(x)\tto \tuple{y,w}
    &
    \mathsf{h}^\mathtt{i}(x) & \to & \tuple{\mathsf{c}(x),\beta_4}
  \end{array}
  \]
  In contrast, the original injectivization transformation would
  return the following system:
  \[
  \begin{array}{r@{~}c@{~}l@{~~~~~~~~~}r@{~}c@{~}l}
  \mathsf{f}^\mathtt{i}(\mathsf{s}(x)) & \to &  \tuple{y,\beta_1(w)} \Leftarrow
  \mathsf{g}^\mathtt{i}(x)\tto \tuple{y,w}  
  &
  \mathsf{g}^\mathtt{i}(x) & \to & \tuple{\mathsf{s}(x),\beta_3}\\
  \mathsf{f}^\mathtt{i}(\mathsf{c}(x)) & \to & \tuple{y,\beta_2(w)} \Leftarrow \mathsf{h}^\mathtt{i}(x)\tto \tuple{y,w}
  &
  \mathsf{h}^\mathtt{i}(x) & \to & \tuple{\mathsf{c}(x),\beta_4}
  \end{array}
  \]
  Finally, the inverse system $\cR_b$ obtained from $\cR_f$ using the
  original transformation has the following form:
  \[
  \begin{array}{r@{~}c@{~}l@{~~~~~~~~~}r@{~}c@{~}l}
    \mathsf{f}^{-1}(y,w) & \to &  \tuple{\mathsf{s}(x)} \Leftarrow
    \mathsf{g}^{-1}(y,w)\tto \tuple{x}  
    &
    \mathsf{g}^{-1}(\mathsf{s}(x),\beta_3) & \to & \tuple{x}\\
    \mathsf{f}^{-1}(y,w) & \to & \tuple{\mathsf{c}(x)} \Leftarrow \mathsf{h}^{-1}(y,w)\tto \tuple{x}
    &
    \mathsf{h}^{-1}(\mathsf{c}(x),\beta_4) & \to & \tuple{x}
  \end{array}
  \]
  For instance, given the forward reduction
  $\mathsf{f}^\mathtt{i}(\mathsf{s}(\mathsf{0})) \to_{\cR_f}
  \tuple{\mathsf{s}(\mathsf{0}),\beta_3}$, we can build the
  corresponding backward reduction:
  $\mathsf{f}^{-1}(\mathsf{s}(\mathsf{0}),\beta_3) \to_{\cR_b}
  \tuple{\mathsf{s}(\mathsf{0})}$.
  
  Note, however, that the left-hand sides of $\mathsf{f}^{-1}$ overlap
  and we should reduce the conditions in order to determine which rule
  to apply. Therefore, in some cases, there is a trade-off between the
  size of the trace terms and the complexity of the reduction steps.
\end{example}
The example above, though, only produces a rather limited improvement
since the considered functions are not recursive. Our next example
shows a much significant improvement. Here, we consider the function
$\mathsf{zip}$ (also used in \cite{MHNHT07} to illustrate the benefits
of an injectivity analysis).

\begin{example}
  Consider the following TRS $\cR$ defining the function
  $\mathsf{zip}$:
  \[
  \begin{array}{r@{~}c@{~}l}
    \mathsf{zip}(\nil,ys) & \to & \nil \\
    \mathsf{zip}(xs,\nil) & \to & \nil \\
    \mathsf{zip}(x:xs,y:ys) & \to & \mathsf{pair}(x,y) :
    \mathsf{zip}(xs,ys) \\
  \end{array}
  \]
  Here, since the third rule is non-erasing, its right-hand side
  contains a single occurrence of a defined function, $\mathsf{zip}$,
  and it does not overlap with any other right-hand side, our improved
  injectivization applies and we get the following \pc\ $\cR_f$:
  \[
  \begin{array}{r@{~}c@{~}l}
    \mathsf{zip}^\mathtt{i}(\nil,ys) & \to & \tuple{\nil,\beta_1(ys)} \\
    \mathsf{zip}^\mathtt{i}(xs,\nil) & \to & \tuple{\nil,\beta_2(xs)} \\
    \mathsf{zip}^\mathtt{i}(x:xs,y:ys) & \to & \tuple{\mathsf{pair}(x,y) :
    zs,w} \Leftarrow \mathsf{zip}^\mathtt{i}(xs,ys)\tto \tuple{zs,w}\\
  \end{array}
  \]
  In contrast, the original injectivization transformation would
  return the following system $\cR'_f$:
  \[
  \begin{array}{r@{~}c@{~}l}
    \mathsf{zip}^\mathtt{i}(\nil,ys) & \to & \tuple{\nil,\beta_1(ys)} \\
    \mathsf{zip}^\mathtt{i}(xs,\nil) & \to & \tuple{\nil,\beta_2(xs)} \\
    \mathsf{zip}^\mathtt{i}(x:xs,y:ys) & \to & \tuple{\mathsf{pair}(x,y) :
    zs,\beta_3(w)} \Leftarrow \mathsf{zip}^\mathtt{i}(xs,ys)\tto \tuple{zs,w}\\
  \end{array}
  \]
  It might seem a small difference, but if we call
  $\mathsf{zip}^\mathtt{i}$ with two lists of $n$ elements, the
  system $\cR'_f$ would build a trace term of the form
  $\beta_3(\ldots \beta_3(\beta_1(\ldots))\ldots)$ with $n$ nested
  constructors $\beta_3$, while $\cR_f$ would just build the trace
  term $\beta_1(\ldots)$. For large values of $n$, this is a significant
  improvement in memory usage.
\end{example}

%%%%%%%%%%%%%%%%%%%%%%%%%%%%%%%%%%%%%%%%%%%%%%%%%%%%%%
\section{Bidirectional Program Transformation} \label{sec:applications}

We illustrate a practical application of our reversibilization
technique in the context of \emph{bidirectional program
  transformation} (see \cite{CFHLST09} for a survey). In particular,
we consider the so-called \emph{view-update} problem. Here, we have a
data structure (e.g., a database) called the \emph{source}, which is
transformed to another data structure, called the
\emph{view}. Typically, we have a \emph{view function},
$\mathsf{view}\!\!: \mathit{Source}\to \mathit{View}$ that takes the
source and returns the corresponding view, together with an
\emph{update} function,
$\mathsf{upd}\!\!: \mathit{View}\times \mathit{Source}\to
\mathit{Source}$ that propagates the changes in a modified view to the
original source. Two basic properties that these functions should
satisfy in order to be well-behaved are the following \cite{FGMPS07}:
\[
\begin{array}{r@{}cl}
  \forall s\in \mathit{Source},\forall v\in \mathit{View}& : & \mathsf{view}(\mathsf{upd}(v,s)) =
  v\\
  \forall s\in \mathit{Source} & : & \mathsf{upd}(\mathsf{view}(s),s) = s
\end{array}
\]
\emph{Bidirectionalization} (first proposed in the database community
\cite{BS81}) basically consists in, given a view function,
``bidirectionalize'' it in order to derive an appropriate update
function. For this purpose, first, a \emph{view complement} function
is usually defined, say $\mathsf{view}^c$, so that the tupled function
\[
\mathsf{view} \vartriangle \mathsf{view}^c\!\! : \mathit{Source} \to
\mathit{View} \times \mathit{Comp}
\]
becomes injective. Therefore, the update function can be defined as
follows:
\[
\mathsf{upd}(v,s) =
(\mathsf{view}\vartriangle\mathsf{view}^c)^{-1}(v,\mathsf{view}^c(s))
\]
This approach has been applied to bidirectionalize view functions in a
functional language in \cite{MHNHT07}. 

In the following, we apply our injectivization and inversion
transformations in order to produce a bidirectionalization
transformation that may be useful in the context of the view-update
problem (with some limitations).
Let us assume that we have a view function, $\mathsf{view}$, that
takes a source and returns the corresponding view, and which is
defined by means of a \pc. Following our approach, given the original
program $\cR$, we produce an injectivized version $\cR_f$ and the
corresponding inverse $\cR_b$. Therefore, in principle, one can use
$\cR_f\cup\cR_b$, which will include the functions
$\mathsf{view}^\mathtt{i}$ and $\mathsf{view}^{-1}\!$, to define an
update function as follows:
\[
\mathsf{upd}(v,s) \to s' 
\Leftarrow \mathsf{view}^\mathtt{i}(s) \tto \tuple{v',\hs},
\mathsf{view}^{-1}\! (v,\hs)\tto\tuple{s'}
\]
where $s$ is the original source, $v$ is the updated view, and $s'$,
the returned value, is the corresponding updated source.  Note that,
in our context, the function $\mathsf{view}^\mathtt{i}$ is somehow
equivalent to $\mathsf{view} \vartriangle \mathsf{view}^c$ above.

Let us now illustrate the bidirectionalization process with an
example. Consider a particular data structure, a list of
\emph{records} of the form $\mathsf{r}(t,v)$ where $t$ is the type of
the record (e.g., $\mathsf{book}$, $\mathsf{dvd}$, $\mathsf{pen}$,
etc.) and $v$ is its price tag.
The following system defines a view function that takes a type and a
list of records, and returns a list with the price tags of the records
of the given type:\footnote{For simplicity, we restrict the record
  types to only $\mathsf{book}$ and $\mathsf{dvd}$.}
\[
\begin{array}{rcl}
  \mathsf{view}(t,\mathsf{nil}) & \to & \mathsf{nil} \\
  \mathsf{view}(t,\mathsf{r}(t',v):rs) & \to &
                                                              \mathsf{val}(\mathsf{r}(t',v)):\mathsf{view}(t,rs) \Leftarrow
                                                              \mathsf{eq}(t,t') \tto \mathsf{true}\\
  \mathsf{view}(t,\mathsf{r}(t',v):rs) & \to &
                                                              \mathsf{view}(t,rs) \Leftarrow
                                                              \mathsf{eq}(t,t') \tto \mathsf{false}\\
  \mathsf{eq}(\mathsf{book},\mathsf{book}) & \to &
                                                   \mathsf{true} \hspace{10ex} 
                                                   \mathsf{eq}(\mathsf{dvd},\mathsf{dvd}) ~ \to~
                                                   \mathsf{true} \\ 
  \mathsf{eq}(\mathsf{book},\mathsf{dvd}) & \to &
                                                  \mathsf{false} \hspace{9.6ex}
                                                  \mathsf{eq}(\mathsf{dvd},\mathsf{book}) ~ \to ~
                                                  \mathsf{false} \\ 
  \mathsf{val}(\mathsf{r}(t,v)) & \to & v\\ 
\end{array}
\]
However, this system is not a \pc. Here, we use a flattening
transformation to produce the following (labeled) \pc\ $\cR$ which is
equivalent for constructor derivations:
\[
\begin{array}{lr@{~}c@{~}l}
  \beta_1: & \mathsf{view}(t,\mathsf{nil}) & \to & \mathsf{nil} \\
  \beta_2: & \mathsf{view}(t,\mathsf{r}(t',v):rs) & \to &
  p:r \\
  &&& \mbox{}\hspace{-10ex}\Leftarrow
  \mathsf{eq}(t,t') \tto \mathsf{true},
  \mathsf{val}(\mathsf{r}(t',v))\tto p,
  \mathsf{view}(t,rs)\tto r\\
  \beta_3: & \mathsf{view}(t,\mathsf{r}(t',v):rs) & \to &
  r \Leftarrow
  \mathsf{eq}(t,t') \tto \mathsf{false},\mathsf{view}(t,rs)\tto r\\[1ex]
  \beta_4: & \mathsf{eq}(\mathsf{book},\mathsf{book}) & \to &
  \mathsf{true} \hspace{10ex}
  \beta_5: ~ \mathsf{eq}(\mathsf{dvd},\mathsf{dvd}) ~ \to ~
  \mathsf{true} \\ 
  \beta_6: & \mathsf{eq}(\mathsf{book},\mathsf{dvd}) & \to &
  \mathsf{false} \hspace{9.5ex}
  \beta_7: ~ \mathsf{eq}(\mathsf{dvd},\mathsf{book}) ~ \to ~
  \mathsf{false} \\[1ex]
  \beta_8: & \mathsf{val}(\mathsf{r}(t,v)) & \to & v\\ 
\end{array}
\]
Now, we can apply our injectivization transformation which returns the
following \pc\ $\cR_f=\mathbf{I}(\cR)$:
\[
\begin{array}{@{}r@{~}c@{~}l@{}}
  \mathsf{view}^\mathtt{i}(t,\mathsf{nil}) & \to & \tuple{\mathsf{nil},\beta_1(t)} \\
  \mathsf{view}^\mathtt{i} (t,\mathsf{r}(t',v):rs) & \to &
  \tuple{p:r,\beta_2(w_1,w_2,w_3)} \\
  &&\mbox{}\hspace{-20ex}\Leftarrow \mathsf{eq}^\mathtt{i}(t,t') \tto \tuple{\mathsf{true},w_1},
  \mathsf{val}^\mathtt{i}(\mathsf{r}(t',v))\tto \tuple{p,w_2},
  \mathsf{view}^\mathtt{i} (t,rs)\tto \tuple{r,w_3}\\
  \mathsf{view}^\mathtt{i} (t,\mathsf{r}(t',v):rs) & \to &
  \tuple{r,\beta_3(v,w_1,w_2)} \\
  &&\mbox{}\hspace{-20ex}\Leftarrow
  \mathsf{eq}^\mathtt{i}(t,t') \tto \tuple{\mathsf{false},w_1},\mathsf{view}^\mathtt{i} (t,rs)\tto \tuple{r,w_2}\\[1ex]
  \mathsf{eq}^\mathtt{i}(\mathsf{book},\mathsf{book}) & \to &
  \tuple{\mathsf{true},\beta_4} \hspace{11ex} 
  \mathsf{eq}^\mathtt{i}(\mathsf{dvd},\mathsf{dvd}) \> \to \>
  \tuple{\mathsf{true},\beta_5} \\ 
  \mathsf{eq}^\mathtt{i}(\mathsf{book},\mathsf{dvd}) & \to &
  \tuple{\mathsf{false},\beta_6} \hspace{9.5ex} 
  \mathsf{eq}^\mathtt{i}(\mathsf{dvd},\mathsf{book}) \> \to \>
  \tuple{\mathsf{false},\beta_7} \\[1ex] 
  \mathsf{val}^\mathtt{i}(\mathsf{r}(t,v)) & \to & \tuple{v,\beta_8(t)}\\ 
\end{array}
\]
Finally, inversion returns the following \pc\
$\cR_b=\mathbf{I}(\cR_f)$:
\[
\begin{array}{@{}r@{~}c@{~}l@{}}
  \mathsf{view}^{-1}\!(\mathsf{nil},\beta_1(t)) & \to & \tuple{t,\mathsf{nil}} \\
  \mathsf{view}^{-1}\! (p:r,\beta_2(w_1,w_2,w_3)) & \to &
  \tuple{t,\mathsf{r}(t',v):rs} \\
  &&\mbox{}\hspace{-30ex}\Leftarrow \mathsf{eq}^{-1}\! (\mathsf{true},w_1) \tto \tuple{t,t'},
  \mathsf{val}^{-1}\! (p,w_2)\tto \tuple{\mathsf{r}(t',v) },\mathsf{view}^{-1}\! (r,w_3)\tto \tuple{t,rs}\\
  \mathsf{view}^{-1}\! (r,\beta_3(v,w_1,w_2)) & \to &
  \tuple{t,\mathsf{r}(t',v):rs } \\
  &&\mbox{}\hspace{-10ex}\Leftarrow
  \mathsf{eq}^{-1}\! (\mathsf{false},w_1) \tto \tuple{t,t'},\mathsf{view}^{-1}\! (r,w_2)\tto \tuple{t,rs }\\[1ex]
  \mathsf{eq}^{-1}\! (\mathsf{true},\beta_4) & \to &
  \tuple{\mathsf{book},\mathsf{book}} \hspace{2.5ex} 
  \mathsf{eq}^{-1}\! (\mathsf{true},\beta_5) \> \to \>
  \tuple{\mathsf{dvd},\mathsf{dvd}} \\ 
  \mathsf{eq}^{-1}(\mathsf{false},\beta_6) & \to &
  \tuple{\mathsf{book},\mathsf{dvd}} \hspace{3.5ex} 
  \mathsf{eq}^{-1}\! (\mathsf{false},\beta_7) \> \to \>
  \tuple{\mathsf{dvd},\mathsf{book}} \\[1ex] 
  \mathsf{val}^{-1}\! (v,\beta_8(t)) & \to & \tuple{\mathsf{r}(t,v) }\\ 
\end{array}
\]
For instance, the term
$\mathsf{view}(\mathsf{book},[\mathsf{r}(\mathsf{book},\mathsf{12}),
\mathsf{r}(\mathsf{dvd},\mathsf{24})])$, reduces to $[\mathsf{12}]$ in
the original system $\cR$. Given a modified view, e.g.,
$[\mathsf{15}]$, we can compute the modified source using function
$\mathsf{upd}$ above:
\[
\mathsf{upd}([\mathsf{r}(\mathsf{book},\mathsf{12}),
\mathsf{r}(\mathsf{dvd},\mathsf{24})],~[\mathsf{15}])
\]
Here, we have the following subcomputations:\footnote{Note that, in
  this case, the function $\mathsf{view}$ requires not only the source
  but also the additional parameter $\mathsf{book}$.}
\[
\begin{array}{l}
\mathsf{view}^\mathtt{i}(\mathsf{book},[\mathsf{r}(\mathsf{book},\mathsf{12}),
\mathsf{r}(\mathsf{dvd},\mathsf{24})])\\
\hspace{20ex}\to_{\cR_f}
\tuple{[\mathsf{12}],\beta_2(\beta_4,\beta_8(\mathsf{book}),
  \beta_3(\mathsf{24},\beta_6,\beta_1(\mathsf{book})))}\\
\mathsf{view}^{-1}\! ([\mathsf{15}],\beta_2(\beta_4,\beta_8(\mathsf{book}),
  \beta_3(\mathsf{24},\beta_6,\beta_1(\mathsf{book})))) \\
\hspace{20ex}\to_{\cR_b}
\tuple{\mathsf{book},[\mathsf{r}(\mathsf{book},\mathsf{15}),
\mathsf{r}(\mathsf{dvd},\mathsf{24})]}
\end{array}
\]
Thus 
$\mathsf{upd}$ returns the updated source
$[\mathsf{r}(\mathsf{book},\mathsf{15}),
\mathsf{r}(\mathsf{dvd},\mathsf{24})]$, as
expected.
We note that the considered example cannot be transformed using the
technique in \cite{MHNHT07}, the closer to our approach, since the
right-hand sides of some rules contain functions which are not
\emph{treeless}.\footnote{A call is \emph{treeless} if it has the form
  $\mathsf{f}(x_1,\ldots,x_n)$ and $x_1,\ldots,x_n$ are different
  variables.} Nevertheless, one could consider a transformation from
\pc\ to functional programs with treeless functions so that the
technique in \cite{MHNHT07} becomes applicable.

Our approach can solve a view-update problem as long as the view
function can be encoded in a \pc. When this is the case, the results
from Section~\ref{sec:transf} guarantee that function $\mathsf{upd}$
is well defined. Formally analyzing the class of view functions that
can be represented with a \pc\ is an interesting topic for further
research.

%%%%%%%%%%%%%%%%%%%%%%%%%%%%%%%%%%%%%%%%%%%%%%%%%%%%%%%%%%%%%%%%%%
\section{Related Work} \label{sec:relwork}

There is no widely accepted notion of reversible computing. In this
work, we have considered one of its most popular definitions,
according to which a computation principle is reversible if there is a
method to \emph{undo} a (forward) computation. Moreover, we expect to
get back to an \emph{exact} past state of the computation. This is
often referred to as \emph{full reversibility}.

As we have mentioned in the introduction, some of the most promising
applications of reversibility include cellular automata \cite{Mor12},
bidirectional program transformation \cite{MHNHT07}, already discussed
in Section~\ref{sec:applications}, reversible debugging \cite{GLM14},
where the ability to go both forward and backward when seeking the
cause of an error can be very useful for the programmer, parallel
discrete event simulation \cite{SJBOQ15}, where reversibility is used
to undo the effects of speculative computations made on a wrong
assumption, quantum computing \cite{Yam14}, where all computations
should be reversible, and so forth.
The interested reader can find detailed surveys in the \emph{state of
  the art} reports of the different working groups of COST Action
IC1405 on Reversible Computation \cite{COST}. 

Intuitively speaking, there are two broad approaches to reversibility
from a programming language perspective: 
\begin{description}
\item[\it Reversible programming languages.] In this case, all
  constructs of the programming language are reversible. One of the
  most popular languages within the first approach is the reversible
  (imperative) language Janus \cite{LD86}. The language was recently
  rediscovered \cite{YAG08b,YAG16,YG07} and has since been formalized
  and further developed.
\item[\it Irreversible programming languages and Landauer's
  embedding.] Alternatively, one can consider an irreversible
  programming language, and enhance the states with some additional
  information (typically, the \emph{history} of the computation so
  far) so that computations become reversible. This is called
  \emph{Landauer's embedding}.
\end{description}
In this work, we consider reversibility in the context of term
rewriting. To the best of our knowledge, we have presented the first
approach to reversibility in term rewriting. A closest approach was
introduced by Abramsky in the context of pattern matching automata
\cite{Abr05}, though his developments could easily be applied to
rewrite systems as well. In Abramsky's approach,
\emph{biorthogonality} was required to ensure reversibility, which
would be a very significant restriction for term rewriting
systems. Basically, biorthogonality requires that, for every pair of
(different) rewrite rules $l\to r$ and $l' \to r'$, $l$ and $l'$ do
not \emph{overlap} (roughly, they do not unify) and $r$ and $r'$ do
not overlap too. Trivially, the functions of a biorthogonal system are
injective and, thus, computations are reversible without the need of a
Landauer embedding. Therefore, Abramsky's work is aimed at defining a
reversible language, in contrast to our approach that is based on
defining a Landauer embedding for standard term rewriting and a
general class of rewrite systems.

Defining a Landauer embedding in order to make a computation mechanism
reversible has been applied in different contexts and computational models, e.g., a
probabilistic guarded command language \cite{Zul01}, a low level
virtual machine \cite{SLZ10}, the call-by-name lambda calculus
\cite{Hue96,Klu99}, cellular automata \cite{Tof77,Mor95}, combinatory
logic \cite{PHW06}, a flowchart language \cite{YAG16}, etc.

In the context of declarative languages, we find the work by Mu
\emph{et al.}~\cite{MHT04}, where a relational reversible language is
presented (in the context of bidirectional programming). A similar
approach was then introduced by Matsuda \emph{et
  al.}~\cite{MHNHT07,MHNHT09} in the context of functional programs
and bidirectional transformation. The functional programs considered
in \cite{MHNHT07} can be seen as linear and \emph{right-treeless}%
\footnote{There are no nested defined symbols in the right-hand sides,
  and, moreover, any term rooted by a defined function in the
  right-hand sides can only take different variables as its proper
  subterms.}  constructor TRSs. The class of functional programs is
more general in \cite{MHNHT09}, which would correspond to
left-linear, right-treeless TRSs.  The reversibilization technique of
\cite{MHNHT07,MHNHT09} includes both an injectivization stage (by
introducing a \emph{view complement} function) and an inversion
stage. These methods are closely related to the transformations of
injectivization and inversion that we have presented in
Section~\ref{sec:transf}, although we developed them from a rather
different starting point. Moreover, their methods for injectivization
and inversion consider a more restricted class of systems than those
considered in this paper. On the other hand, they apply a number of
analyses to improve the result, which explains the smaller traces in
their approach.  All in all, we consider that our approach gives
better insights to understand the need for some of the requirements of
the program transformations and the class of considered programs. For
instance, most of our requirements come from the need to remove
programs positions from the traces, as shown in Section~\ref{sec:irr}.

Finally, \cite{TA15} considers the reversible language RFUN. Similarly
to Janus, computations in RFUN are reversible without the need of a
Landauer embedding. The paper also presents a transformation from a
simple (irreversible) functional language, FUN, to RFUN, in order to
highlight how irreversibilities are handled in RFUN. The
transformation has some similarities with both the approach of
\cite{MHNHT07} and our improved transformation in
Section~\ref{sec:improved}; on the other hand, though, \cite{TA15}
also applies the Bennett \emph{trick} \cite{Ben73} in order to avoid
some unnecessary information.

%%%%%%%%%%%%%%%%%%%%%%%%%%%%%%%%%%%%%%%%%%%%%%%%%%%%%%%%%%%%%%%%%%
\section{Discussion and Future Work} \label{sec:conclusion}

In this paper, we have introduced a reversible extension of term
rewriting. In order to keep our approach as general as possible, we
have initially considered DCTRSs as input systems, and proved the
soundness and reversibility of our extension of rewriting. Then, in
order to introduce a reversibilization transformation for these
systems, we have also presented a transformation from DCTRSs to pure
constructor systems (\pcs) which is correct for constructor
reduction. A further improvement is presented for injective functions,
which may have a significant impact in memory usage in some
cases. Finally, we have successfully applied our approach in the
context of bidirectional program transformation.

We have developed a prototype implementation of the reversibilization
transformations introduced in Section~\ref{sec:transf}.  The tool can
read an input TRS file (format \texttt{.trs} \cite{termcom}) and then
it applies in a sequential way the following transformations:
flattening, simplification of constructor conditions, injectivization,
and inversion. The tool prints out the CTRSs obtained at each
transformation step. It is publicly available through a web interface
from \texttt{http://kaz.dsic.upv.es/rev-rewriting.html}, where we have included
a number of examples to easily test the tool.

As for future work, we plan to investigate new methods to further
reduce the size of the traces. In particular, we find it interesting
to define a reachability analysis for DCTRSs. A reachability analysis
for CTRSs without extra-variables (1-CTRSs) can be found in
\cite{FG03}, but the extension to deal with extra-variables in DCTRSs
(since a DCTRS is a particular case of 3-CTRS) seems
challenging. Furthermore, as mentioned in the paper, a completion
procedure to add \emph{default} cases to some functions (as suggested
in Section~\ref{sec:injectivization}) may help to broaden the
applicability of the technique and avoid the restriction to
constructor reduction. Finally, our injectivization and inversion
transformations are correct w.r.t.\ innermost reduction. Extending our
results to a lazy strategy is also an interesting topic for further
research.

\subsection*{Acknowledgements}

We thank the anonymous reviewers for their useful comments and suggestions to improve
this paper.

%%%%%%%%%%%%%%%%%%%%%%%%%%%%%%%%%%%%%%%%%%%%%%%%%%%%%%%%%%%%%%%%%%
%\bibliographystyle{abbrv}
%\bibliography{biblio} 

\end{document}